\documentclass[aip,jmp,amsmath,amssymb,preprint,]{revtex4-1}
\usepackage{amssymb,amsmath,bm}

\usepackage{color}
\usepackage{mathrsfs,amsmath}
\usepackage{mathrsfs}
\usepackage{graphicx}
\usepackage{dcolumn}
\usepackage{bm}
\usepackage{natbib}
\usepackage{amsthm}

\begin{document}
\newtheorem{thm}{Theorem}[section]
\newtheorem{lemma}[thm]{Lemma}
\newtheorem{prop}[thm]{Proposition}
\newtheorem{rem}[thm]{Remark}
\newtheorem{cor}[thm]{Corollary}


\title{Conservation-Dissipation Formalism for Soft Matter Physics: II. Application to Non-isothermal Nematic Liquid Crystals} 
\author{Liangrong Peng}%
\affiliation{Zhou Pei-Yuan Center for Applied Mathematics, Tsinghua University, Beijing, China, 100084}
\author{Yucheng Hu}%
\affiliation{Zhou Pei-Yuan Center for Applied Mathematics, Tsinghua University, Beijing, China, 100084}
\author{Liu Hong}
\homepage{Author to whom correspondence should be addressed. Electronic mail: zcamhl@tsinghua.edu.cn}
\affiliation{Zhou Pei-Yuan Center for Applied Mathematics, Tsinghua University, Beijing, China, 100084}
\date{\today} 

\begin{abstract}
To most existing non-equilibrium theories, the modeling of non-isothermal processes was a hard task.
Intrinsic difficulties involved the non-equilibrium temperature, the coexistence of conserved energy and dissipative entropy, etc.
In this paper, by taking the non-isothermal flow of nematic liquid crystals as a typical example, we illustrated that thermodynamically consistent models in either vectorial or tensorial forms could be constructed within the framework of Conservation-Dissipation Formalism (CDF).
And the classical isothermal Ericksen-Leslie model and Qian-Sheng model were shown to be special cases of our new vectorial and tensorial models in the isothermal, incompressible and stationary limit.
Most importantly, from above examples, it was learnt that mathematical modeling based on CDF could easily solve the issues relating with non-isothermal situations in a systematic way.
The first and second laws of thermodynamics were satisfied simultaneously.
The non-equilibrium temperature was defined self-consistently through the partial derivative of entropy function.
Relaxation-type constitutive relations were constructed, which gave rise to the classical linear constitutive relations, like Newton's law and Fourier's law, in stationary limits.
Therefore, CDF was expected to have a broad scope of applications in soft matter physics, especially under the complicated situations, such as non-isothermal, compressible and nanoscale systems.
\end{abstract}


\maketitle

\section{Introduction}
Liquid crystals were intermediate states of materials between solid crystals and liquid.
According to molecular symmetry, liquid crystals were loosely classified into nematics, cholesterics and smectics.
Among them, the nematic liquid crystals, which were often composed of long, thin, rod-like molecules with long axes of neighbouring molecules aligned roughly parallel to each other, were the most studied phase \cite{deGennes1995The}.

Based on scales of description, theories for nematic liquid crystals could be categorized into three different but closely-related types --
the Doi-Onsager theory, Landau-de Gennes theory and Ericksen-Leslie theory.
The first one belonged to a molecular model, which focused on the distribution function of molecular orientations.
The second one was a tensorial model, in which a symmetric traceless tensor ($Q$-tensor) was utilized to describe the orientational properties of molecules.
The last one was a vectorial model, which used a unit vector to characterize the average direction of molecules \cite{Han2015From}.
Notice that the second moment of the orientational distribution function corresponded to the $Q$-tensor, while the uniaxial form of $Q$-tensor reduced to the director vector in the Ericksen-Leslie theory.
As a consequence, models in different scales were connected to each other under certain limit processes.
It was rigorously shown in mathematics that \cite{Han2015From, Wang2015Rigorous}, the Doi-Onsager equation could be simplified to the $Q$-tensor theory by Bingham closure and Taylor expansion, and the Ericksen-Leslie equation can be derived from the $Q$-tensor theory by expansion near the local equilibrium.

Pioneering works on the macroscopic continuum modeling of nematic liquid crystals could be dated back to Ericksen \cite{Ericksen1962Hydrostatic} and Leslie \cite{Leslie1968Some, Leslie1979Theory} in the 1960s, who established rather general conservation laws and constitutive equations from a hydrodynamical point of view.
The Ericksen-Leslie (E-L) theory was shown to be successful in modeling and explaining many intrinsic physical phenomena of the nematics, including the
flow alignment, electric and magnetic induced flows in display devices
\cite{Leslie1979Theory, deGennes1995The, Chandrasekhar1992Liquid} and \textit{etc.}

Due to the complexity of E-L theory in mathematical analysis, Lin and Liu \cite{Lin1995Nonparabolic, Lin2000Existence} tried to simplify it by introducing a penalty approximation to relax the nonlinear constraint of the director and reducing the Oseen-Frank energy into a single term.
The global existence of weak solutions of the simplified isothermal model was proved rigorously.
Later, a similar model was derived by Sun and Liu \cite{Sun2008On} within a general energetic variational framework, and by us with the Conservation-Dissipation Formalism in the first paper of this series \cite{Peng2018Conservation}.
For recent developments on mathematical models for the hydrodynamic flows of nematic liquid crystals, see refs \cite{Lin2001Static, Fanghua2014Recent} for details.

The classical E-L theory presumed the liquid crystal was at constant temperature and the director was of constant length. However, the non-isothermal and compressible nematic liquid crystals were far more interesting and important in industrial applications. A key reason was that for nematics of thermotropic type the most usual operation to induce a phase transition was to change the temperature. Furthermore, mixtures of different nematics were widely used in the display industry in order to obtain a low melting point \cite{deGennes1995The}.
In contrast to isothermal E-L models, results for non-isothermal cases were not so fruitful.
The generalized E-L model in non-isothermal situations was presented,
by Hieber and Pr\"{u}ss \cite{Hieber2016Thermodynamical, Hieber2017Dynamics} based on the principle of thermodynamical consistency,
in the incompressible case by Feireisl, Rocca and Schimperna \cite{Feireisl2011On},
and by De Anna and Liu \cite{Anna2017Non} with a generalized Oseen-Frank energy.

On the other hand, one could discuss the hydrodynamic behaviour of liquid crystal flows based on the Landau-de Gennes (LdG) theory \cite{deGennes1995The}.
The LdG theory shared the same variables of density and velocity as the E-L theory. A fundamental difference between them lied on the choice of order parameter.
To account for states of nematics at time $t$ and position $x$, the E-L theory adopted the director vector $d(x,t)$, which was sufficient to describe the nematic dynamics of low-molar mass, while the LdG theory considered a traceless symmetric tensor $Q(x,t)$ that was more accurate in the presence of high disclination density \cite{Gay2011The}.
As to the $Q$-tensor theory, there were plenty of different models, based on the orientational distribution function \cite{Han2015From}, derived from variational principles \cite{Beris1994Thermodynamics, Qian1998Generalized}, and recently for non-isothermal nematics \cite{Feireisl2014Evolution, Feireisl2015Nonisothermal}.
However, a general dynamical theory that took non-isothermal and compressible effects into account for liquid crystals described by $Q$-tensor was still lacking \cite{Sonnet2004Continuum}.

In this paper, we mainly concentrated on applications of the Conservation-Dissipation Formalism (CDF) to the mathematical modeling of non-isothermal flows of nematic liquid crystals. New generalized vectorial (E-L) theory and tensorial ($Q$-tensor) theory for non-isothermal and compressible situations were constructed. Concrete expressions of entropy flux and entropy production rate of the system, as well as constitutive relations for the stress tensor, director/$Q$-tensor body torque, director/$Q$-tensor surface torque and heat flux were derived. Especially, under the isothermal and incompressible condition, the classical E-L theory or Qian-Sheng model were shown to be special cases of our vectorial and tensorial models in the stationary limit.

\section{Vectorial models for non-isothermal flows of nematic liquid crystals}
\label{E-L theory}

\subsection{Conservation laws}

In this part, we considered non-isothermal flows of liquid crystals in the nematic phase.
For simplicity, external electric and magnetic fields were not taken into consideration.
The conservation laws of mass, momentum, angular momentum and total energy took the following form:
\begin{align}
&\frac{\partial}{\partial t}\rho + \nabla \cdot (\rho v)=0,\label{mass}\\
&\frac{\partial}{\partial t}(\rho v)+ \nabla\cdot(\rho v \otimes v)=\xi + \nabla \cdot \sigma,\label{linmom}\\
&\frac{\partial}{\partial t}(\rho_1 w)+ \nabla\cdot(\rho_1 v \otimes w)= g + \nabla \cdot \pi, \label{angmom}\\
&\frac{\partial}{\partial t}(\rho e) + \nabla \cdot(\rho v e)=\xi \cdot v + \nabla \cdot (\sigma \cdot v + \pi \cdot w - q)\label{energy}.
\end{align}
Here $\rho$ was the density of liquid crystals, $v \in \mathbb R^3$ was the velocity vector, $\xi$ was the external body force per unit volume, $\sigma$ was the stress tensor.
In Eq. \eqref{angmom}, $\rho_1=\rho |r|^2$ was the density of inertia moment,
where $r$ was the effective position vector. Its norm $|r|$ was assumed to be constant in what follows.
$d\in \mathbb R^3$ was the director vector which described the preferred orientation of molecules. The material derivative of $d$ gave the director velocity $w \equiv \frac{d}{dt}(d)$. 
$\pi$ was the director surface torque,
$g$ was the intrinsic body torque acting on the director. Both were measured per unit volume, describing effects of macroscopic flows on the microscopic structure.
In Eq. \eqref{energy}, $e=u+ \frac{1}{2}v^2+ \frac{\rho_1}{2\rho}w^2$ denoted the specific total energy density, including both translational and rotational kinetic energies as $\frac{1}{2}v^2$ and $\frac{\rho_1}{2\rho}w^2$, and the internal energy $u$.
$q$ represented the heat flux.
The symbol $\otimes$ denoted the tensor product, $(a \otimes b)_{ij}=a_i b_j$ for $a, b\in \mathbb R^3$.

In general, the stress tensor could be decomposed into four parts based on their different origins,
\begin{equation}\label{sigma_EL}
\sigma= -pI+ \sigma_V + \sigma_E+ \sigma_L,
\end{equation}
where $p, \sigma_V, \sigma_E, \sigma_L$ denoted the thermodynamic pressure, viscous stress, Ericksen stress for the static state and Leslie stress for the non-equilibrium state, respectively. From the perspective of viscous-elastic fluids, $\sigma_E$ was induced by elastic distortions,
$(\sigma_V+\sigma_L)$ together composed the stress induced by viscosity effects, in which $\sigma_V$ was the viscous stress for homogenous fluid flows, and $\sigma_L$ was the orientation-induced viscous stress.
A similar decomposition was introduced by Hieber and Pr\"{u}ss \cite{Hieber2017Dynamics}.
However, in their work, constitutive relations of stresses were given directly and without a clear physical interpretation.
Analogously, the director surface torque $\pi$ and director body torque $g$ could be separated as
\begin{equation}
\pi=\pi_V+\pi_E+\pi_L,\quad g=g_V+g_E+g_L,
\end{equation}
with the subscript $V, E, L$ denoting the viscous part, Ericksen part and Leslie part,  respectively.

\begin{rem}
Note in many cases the left-hand side of Eq. \eqref{angmom} was neglected for simplicity.
However, when the anisotropic axis was subjected to large accelerations \cite{Leslie1979Theory,Stewart2004The}, the director inertial term would play
an important role.
And it was kept in the current study to preserve the elegant mathematical structure of local conservation laws.
\end{rem}

\begin{rem}
In most previous works, the director vector $d$ was assumed to be uniform and was simplified as a unit vector in a unit sphere $\mathbb {S}^2$.
It was necessary to emphasize that, here we introduced $d\in \mathbb {R}^3$ as a 3-dimensional vector, to account for both the preferred orientation and the average length \cite{Flory1984Molecular} of rod-like molecules.
In this way, the vectorial theory could be extended to the case of mixtures of molecules with varied lengths \cite{Lekkerkerker1984On, Flory1984Molecular, He2016Isotropic}.
\end{rem}

\begin{rem}
In Eqs. \eqref{mass}-\eqref{energy}, local conservation laws of $(\rho, v, w, e)$ were directly introduced.
Physically, they were a straightforward consequence of the first law of thermodynamics (or the energy conservation law) and the Galileo principle \cite{jou1996extended}.
\end{rem}

The conservation of total energy Eq. \eqref{energy} could also be interpreted as the first law of thermodynamics, which stated that the internal energy was changed by either doing work or exchanging heat with the environment.
The following proposition gave an explicit expression of the first law of thermodynamics for the non-isothermal flows of liquid crystals, which was first considered by Leslie \cite{Leslie1968Some} from rational thermodynamics.

\begin{prop}
Assuming that Eq. \eqref{energy} satisfied the principle of material frame-indifference,
then the first law of thermodynamics stated that, the change in internal energy resulted from the work done to the system and the heat exchanged with its surroundings:
\begin{equation} \label{internalenergy}
\rho\frac{du}{dt}
={\mathcal W} - \nabla \cdot q ,
\end{equation}
where the work ${\mathcal W}= \sigma^T: A + \pi^T : M - g \cdot N$ adopted a bilinear structure as the product of generalized forces $\sigma, \pi, (-g)$ and displacements $A, M, N$, and $q$ described the transportation of internal energy through heat flux.
Here $2A_{ij}=v_{i,j}+ v_{j,i}, ~2 \Omega_{ij}=v_{i,j} - v_{j,i},~N_i=w_i- \Omega_{ik} d_k,~  M_{ij}= w_{j,i}+ \Omega_{kj} d_{k,i}$.
\end{prop}

\begin{proof}
By taking scalar product of Eq. \eqref{linmom} with $v$ and Eq. \eqref{angmom} with $w$, and using the continuity equation, Eq. \eqref{energy} reduced to the balance law of internal energy
\begin{equation}\label{internalenergy0}
  \rho\frac{du}{dt}=\sigma^T : \nabla v + \pi^T : \nabla w - g \cdot w - \nabla \cdot q,
\end{equation}
where the superscript $T$ meant transposition, the colon $:$ was the double inner product between two second-order tensors, \textit{i.e.}, $A:B=\sum_{i,j} A_{ij}B_{ji}$, and $(\nabla v)_{ij} = \partial v_j/ \partial x_i= v_{j,i}$.

Notice that quantities $\rho, u, q$ in Eq. \eqref{internalenergy0} were objective and frame-indifference, while $\nabla v$, $\nabla w$ and $w$ were not.
To satisfy the principle of material frame-indifference, following Leslie \cite{Leslie1968Some}, the right-hand side of Eq. \eqref{internalenergy0} could be rearranged into a convenient form:
\begin{align}
  \rho\frac{du}{dt}
  =&\sigma_{ij} v_{j,i} + \pi_{ij} w_{j,i} - g_i w_i  - \partial_i q_i \nonumber \\
  =&\sigma_{ij} (A_{ij} -\Omega_{ij}) + \pi_{ij} (M_{ij} - \Omega_{kj} d_{k,i})
     - g_i (N_i + \Omega_{ik}d_k) - \partial_i q_i \nonumber \\
  =&( \sigma_{ij} A_{ij}  + \pi_{ij} M_{ij} - g_i N_i- \partial_i q_i )
    +(- \sigma_{ij} \Omega_{ij}   - \pi_{ij} \Omega_{kj} d_{k,i} - g_i \Omega_{ik}d_k) \nonumber \\
  =& \sigma^T : A + \pi^T :M - g \cdot N - \nabla \cdot q +
  \underline{[\sigma - \pi^T \cdot (\nabla d) + g \otimes d ]:\Omega} \nonumber \\
  \equiv
  & \sigma^T : A + \pi^T :M - g \cdot N - \nabla \cdot q + \tilde \sigma : \Omega ,  \label{internalenergy2}
\end{align}
where in the fourth step we used relations $\Omega_{ij}=-\Omega_{ji}$ and $\pi_{ij} \Omega_{kj} d_{k,i}=(\pi^T)_{ji} (\nabla d)_{ik} \Omega_{kj}$ with symmetric and skew-symmetric parts of the velocity gradient defined as $2A_{ij}=v_{i,j}+ v_{j,i}$ and $2 \Omega_{ij}=v_{i,j} - v_{j,i}$.

In above equations, $(A, M, N)$ were objective. $N$ was the relative angular velocity, representing the change rate of director with respect to background flows \cite{deGennes1995The}.
$N$ was also known as the co-rotational derivative of the director vector $d$, which was measured by an observer whose reference was carried by fluids and rotated with fluids \cite{jou1996extended}.
Since Eq. \eqref{internalenergy2} satisfied the principle of material frame-indifference \cite{Leslie1968Some}, the underlined term
$\tilde \sigma : \Omega=[\sigma - \pi^T \cdot (\nabla d) + g \otimes d]:\Omega=0$
and vanished in the following derivation.
\end{proof}

\subsection{Generalized Gibbs relation}
To close the system of partial differential equations in \eqref{mass}-\eqref{angmom} and \eqref{internalenergy}, constitutive relations for $(\sigma,\pi, g, q)$ were needed.
Following the Conservation-Dissipation Formalism \cite{zhu2015conservation}, we introduced a strictly concave mathematical entropy function
\begin{equation}
\eta = \rho s(\nu, u, d, \nabla d, C, K, l,h).
\end{equation}
Here $\nu=1/\rho$ was the specific volume, $(C, K)$ were tensors with the same size of $(\sigma,\pi)$, and $(l,h)$ were vectors with the same size of $(g,q)$.
$(C, K, l)$ were used to describe the viscous-elastic effects of nematic liquid crystal flows, and $h$ characterized the heat conduction induced by temperature gradients.
Notice that $(\nu, u, d, \nabla d)$ represented conserved variables, while $(C, K, l,h)$ represented dissipative variables, which would be specified later.
In equilibrium, $C=K=0, l=h=0$, 
we had
\begin{equation}
s|_{C=K=0,~ l=h=0} \equiv s_0(\nu, u, d, \nabla d),
\end{equation}
where $s_0$ was the equilibrium entropy.

With the entropy function in hand, the non-equilibrium temperature $\theta$ and thermodynamic pressure $p$ were defined by
\begin{equation}
\theta^{-1}=\frac{\partial s}{\partial u}, \quad \theta^{-1}p=\frac{\partial s}{\partial \nu}.
\end{equation}
Consequently, the equilibrium temperature was given by $T=(\frac{\partial s_0}{\partial u})^{-1}$.

In the next, we were going to examine the time evolution of the entropy. And a manifestation of the second law of thermodynamics would be shown.
Firstly, we assumed the entropy $s(\nu, u, d, \nabla d, C, K, l,h)$ to be an isotropic function of $d$ and $\nabla d$,
or in other words it was objective with respect to conserved variables.
As a direct consequence, the tensor $(
\frac{\partial s}{\partial d}\otimes d +  \frac{\partial s}{\partial \nabla d} \cdot \nabla d - \nabla d \cdot  \frac{\partial s}{\partial \nabla d}
)$ was symmetric. This result was first pointed out by Ericksen \cite{Ericksen1961Conservation} and Leslie \cite{Leslie1968Some} in the 1960s, and we represented it as Lemma \ref{tensor symmetric} in Appendix to maintain the self-integrity of formulation.

Secondly,
to simplify notations, we employed a differential operator $\mathcal D$ as $\mathcal D s=(\rho s)_t +\nabla \cdot (v \rho s)$.
Utilizing the continuity equation \eqref{mass}, one deduced that $\mathcal Ds=\rho {ds}/{dt}$.
Then, based on the conservation laws \eqref{mass}-\eqref{angmom} and balance law \eqref{internalenergy}, it was direct to calculate the evolution equation of the entropy function as follows:
\begin{align*}
 \eta_t + \nabla \cdot (v \eta)
\equiv& \mathcal D s(\nu, u, d, \nabla d, C, K, l,h)\\
=&\frac{\partial s}{\partial \nu} \mathcal D\nu + \frac{\partial s}{\partial u} \mathcal D u + \frac{\partial s}{\partial d} \cdot  \mathcal D d
+(\frac{\partial s}{\partial \nabla d})^T:  \mathcal D(\nabla d) \\
&+(\frac{\partial s}{\partial C})^T: \mathcal D C + (\frac{\partial s}{\partial K})^T: \mathcal D K + \frac{\partial s}{\partial l} \cdot
\mathcal D l + \frac{\partial s}{\partial h}\cdot  \mathcal D h \\
=& \theta^{-1} [(pI+\sigma^T): A + \pi^T : M - g \cdot N - \nabla \cdot q ]
+\frac{\partial s}{\partial d} \cdot  \mathcal D d
+(\frac{\partial s}{\partial \nabla d})^T:  \mathcal D(\nabla d) \\
&+(\frac{\partial s}{\partial C})^T: \mathcal D C + (\frac{\partial s}{\partial K})^T: \mathcal D K + \frac{\partial s}{\partial l}\cdot  \mathcal D l + \frac{\partial s}{\partial h}\cdot  \mathcal D h.
\end{align*}

By making use of the following relations (See Lemma \ref{lemma of Dd} in Appendix)
\begin{equation} \label{Dd}
\mathcal D d= \rho(N+\Omega \cdot d), \quad
\mathcal D (\nabla d) = \rho[ M -(\nabla d) \cdot \Omega + \Omega \cdot(\nabla d)-A \cdot (\nabla d)],
\end{equation}
and noticing $A^T=A$, $\Omega^T=-\Omega$, we had
\begin{align*}
&\frac{\partial s}{\partial d} \cdot  \mathcal D d + (\frac{\partial s}{\partial \nabla d})^T:  \mathcal D(\nabla d)
=\rho \frac{\partial s}{\partial d_i} w_i + \rho \frac{\partial s}{\partial d_{i,j}} \frac{d}{dt} {d_{i,j}}
\\
=&\rho \frac{\partial s}{\partial d_i} (N_i + \Omega_{ik} d_k) +
\rho \frac{\partial s}{\partial d_{i,j}} (M_{ji}-d_{k,j} \Omega_{ki} + \Omega_{jk}d_{i,k}-A_{jk} d_{i,k}) \\
=&\rho \frac{\partial s}{\partial d_i} N_i
+ \rho \frac{\partial s}{\partial d_{i,j}} (M_{ji}-A_{jk} d_{i,k})
+ \rho \frac{\partial s}{\partial d_i} \Omega_{ik} d_k
+ \rho \frac{\partial s}{\partial d_{i,j}} (-d_{k,j} \Omega_{ki} + \Omega_{jk}d_{i,k})\\
=&\rho \frac{\partial s}{\partial d} \cdot N
+ \rho \frac{\partial s}{\partial \nabla d} :  M
- \rho (\nabla d \cdot \frac{\partial s}{\partial \nabla d}) :  A
-\underline{
\rho (
\frac{\partial s}{\partial d}\otimes d +  \frac{\partial s}{\partial \nabla d} \cdot \nabla d - \nabla d \cdot  \frac{\partial s}{\partial \nabla d}
) :\Omega
} \\
=&\rho \frac{\partial s}{\partial d} \cdot N + \rho \frac{\partial s}{\partial \nabla d} :  M - \rho ( \nabla d \cdot \frac{\partial s}{\partial \nabla d} ) :  A .
\end{align*}
The underlined term
in the fourth step vanished, since $\Omega$ was anti-symmetric while
$(\frac{\partial s}{\partial d}\otimes d +  \frac{\partial s}{\partial \nabla d} \cdot \nabla d - \nabla d \cdot  \frac{\partial s}{\partial \nabla d} )$ was symmetric by Lemma \ref{tensor symmetric}.
Substituting the above formula into $\mathcal D s$, we arrived at
\begin{equation}
\begin{split}
\mathcal D s
=&
\nabla \cdot ({-\theta^{-1}q}) + \theta^{-1}(pI+\sigma^T): A - \rho ( \nabla d \cdot \frac{\partial s}{\partial \nabla d} ):A
+\theta^{-1} \pi^T : M  + \rho \frac{\partial s}{\partial \nabla d} :  M  \\
& -\theta^{-1} g \cdot N + \rho \frac{\partial s}{\partial d} \cdot N + q\cdot \nabla \theta^{-1}
+s_C^T: \mathcal D C + s_K^T: \mathcal D K + s_l \cdot \mathcal D l + s_h \cdot \mathcal D h \\
=&-\nabla \cdot ({\theta^{-1}q})
+\underline{
[\theta^{-1}\sigma_E^T  - \rho ( \nabla d \cdot \frac{\partial s}{\partial \nabla d} )]:A + (\theta^{-1} \pi_E^T + \rho \frac{\partial s}{\partial \nabla d}):  M  - (\theta^{-1} g_E - \rho \frac{\partial s}{\partial d}) \cdot N
}  \\
&+\underline{
( s_C^T:\mathcal D C+ \theta^{-1}\sigma_V^T:A )}
+\underline{
(s_K^T: \mathcal D K + \theta^{-1} \pi_V^T : M)}
+\underline{
(s_l \cdot  \mathcal D l-\theta^{-1} g_V \cdot N)}
 \\
&
+\underline{
(s_h \cdot  \mathcal D h + q\cdot \nabla \theta^{-1})}
+\underline{
\theta^{-1}(\sigma_L^T:A + \pi_L^T:M - g_L \cdot N)
}
  \\
\equiv & \nabla \cdot J^f +epr0+ epr1+ epr2+  epr3+ epr4+ epr5,
\end{split}
\end{equation}
where decompositions of $\sigma= -pI+ \sigma_V + \sigma_E+ \sigma_L$, $\pi=\pi_V+\pi_E+\pi_L$, and $g=g_V+g_E+g_L$ were utilized in the second step.
The entropy flux was $J^f=-\theta^{-1}q$, and the entropy production rate $\Sigma^f \geq 0$ included six separate contributions, $epr0, epr1, \cdots, epr5$, which would be discussed in detail later.

\subsection{Decomposition of entropy production rate}

Notice that the decomposition of the entropy production rate $\Sigma^f$ was not unique.
In this paper, we adopted a straightforward decomposition mainly based on different physical origins of variables, including the Ericksen (static) part, Leslie (non-equilibrium) part and viscous part.
Meanwhile, we specified $(\sigma_V,\pi_V, g_V, q)$ to be conjugate variables of $(C,K,l,h)$ with respect to the specific entropy function $s$, respectively.

Recall that, the subscript $E$ denoted the static deformation value, $i.e.$, the Ericksen part, which as a consequence made no contribution to the entropy production rate:
\begin{equation}\label{epr0}
epr0={
[\theta^{-1}\sigma_E^T  - \rho ( \nabla d \cdot \frac{\partial s}{\partial \nabla d} )]:A + (\theta^{-1} \pi_E^T + \rho \frac{\partial s}{\partial \nabla d}):  M  - (\theta^{-1} g_E - \rho \frac{\partial s}{\partial d}) \cdot N
}=0.
\end{equation}
Due to the arbitrariness of $A, M$ and $N$, the Ericksen parts of the stress tensor, director surface torque, and director body torque became
\begin{equation}\label{ericksen}
\sigma_E=  \rho \theta (\nabla d \cdot \frac{\partial s}{\partial \nabla d})^T,\quad
\pi_E   = -\rho \theta (\frac{\partial s}{\partial \nabla d})^T,\quad
g_E     =  \rho \theta \frac{\partial s}{\partial d},
\end{equation}
respectively.

To guarantee the second law of thermodynamics, we further assumed the entropy production rates from the classical viscous part and Leslie part were both non-negative,
\begin{align*}
epr1+ epr2+ epr3+ epr4
=& {( s_C^T:\mathcal D C+ \theta^{-1}\sigma_V^T:A )} +
(s_K^T: \mathcal D K + \theta^{-1} \pi_V^T : M)  \\
&+(s_l \cdot  \mathcal D l-\theta^{-1} g_V \cdot N)
 +(s_h \cdot  \mathcal D h + q\cdot \nabla \theta^{-1}) \geq 0,
\end{align*}
and $epr5 \geq 0$.
Given the fact that $\sigma_V$, $\pi_V$, $g_V$ and $q$ were conjugate variables of $C$, $K$, $l$ and $h$ with respect to the specific entropy $s$, we had
\begin{equation}
\frac{\partial s}{\partial C} =\sigma_V, \quad  \frac{\partial s}{\partial K} =\pi_V, \quad \frac{\partial s}{\partial l}=g_V, \quad \frac{\partial s}{\partial h}=q.
\end{equation}
Having the expression of entropy production rate in hand, CDF \cite{zhu2015conservation} suggested following constitutive equations:
\begin{equation}
    \begin{pmatrix}
       (\rho C)_t+\nabla\cdot(\rho v \otimes C) + \theta^{-1}A\\
       (\rho K)_t+\nabla\cdot(\rho v \otimes K) + \theta^{-1}M\\
       (\rho l)_t+\nabla\cdot(\rho v \otimes l) -\theta^{-1} N\\
       (\rho h)_t+\nabla\cdot(\rho v \otimes h) +\nabla \theta^{-1}
    \end{pmatrix}
		={\Upsilon} \cdot
     \begin{pmatrix}
      s_C \\ s_K \\s_l\\ s_h
     \end{pmatrix},
\end{equation}
where the nonlinear dissipation matrix ${\Upsilon}={\Upsilon}(\nu, u, d, \nabla d, C, K, l,h)$ was positive definite, and was readily dependent on both conserved and dissipative variables.

Moreover, the remaining Leslie part of the entropy production rate was
\begin{equation}\label{epr_1}
epr{5} \equiv \theta^{-1}(\sigma_L^T:A + \pi_L^T:M - g_L\cdot N)\geq0.
\end{equation}
This was exactly the total entropy production rate in the classical E-L theory, except that the dissipation caused by the director surface torque $\pi_L^T:M$ was also included (See Eq. (3.1.21) in ref. \cite{Chandrasekhar1992Liquid}).
Specifically, by choosing $\pi_L=0$ and following the same argument of Leslie \cite{Leslie1968Some,Chandrasekhar1992Liquid}, one would arrive at the most widely adopted form of $\sigma_L$ and $g_L$ as follows:
\begin{align}
 \sigma_L =&\alpha_1 (d^T \cdot A \cdot d) d \otimes d  + \alpha_2 N \otimes d  +\alpha_3 d \otimes N
+ \alpha_4 A  +\alpha_5 d \otimes  (A \cdot d)  +\alpha_6 (A \cdot d) \otimes d, \\
g_L=&(\alpha_2-\alpha_3)N + (\alpha_5- \alpha_6)A \cdot d,
\end{align}
where
$\alpha_i=\alpha_i(\rho, \theta)~(i=1,\cdots,6)$ was called the classical Leslie coefficient, depending on the density $\rho$ and non-equilibrium temperature $\theta$.
According to Parodi's relation \cite{Chandrasekhar1992Liquid}, these coefficients were connected by $\alpha_2 + \alpha_3=\alpha_6 -\alpha_5$, and the number of independent coefficients reduced to five.
In general, the Leslie coefficients could not be identified separately, except for $\alpha_4$. Fortunately, they were
combinations of Miesowicz viscosities, which could be measured accurately.
We referred to ref. \cite{Stewart2004The} for the physical interpretation of Leslie coefficients and their relations with experimentally measurable quantities.

By substituting above formulas into $epr5$, one had
\begin{align*}
\theta\cdot epr{5}
=& \alpha_1 (d^T \cdot A \cdot d)^2 + (\alpha_2 + \alpha_3 - \alpha_5 + \alpha_6) (d^T \cdot A \cdot N)  + \alpha_4 |A|^2  \\
&+(\alpha_5 +\alpha_6) |A \cdot d|^2 - (\alpha_2-\alpha_3) |N|^2,
\end{align*}
where $A^T=A$ was used, and $|A|^2\equiv A^T:A$, $|N|^2 \equiv N^T \cdot N$.
Since variables $(d, A, N)$ were arbitrary, we had to put restrictions on the Leslie coefficients in order to guarantee the non-negativity of entropy production rate $epr5$, which was stated in the following.
\begin{thm}\label{2nd law}
The Leslie contribution to the entropy production rate in $epr5$ was non-negative if and only if
\begin{equation}\label{ineq of 2nd law}
\begin{aligned}
& \alpha_4 \geq 0, \quad \alpha_5 + \alpha_6\geq 0, \\
&(\alpha_1+\alpha_5+\alpha_6)d_1^2 +\alpha_4 \geq 0,\\
& 4(\alpha_3 -\alpha_2)(\alpha_1+\alpha_5+\alpha_6)d_1^2 + 4(\alpha_3 -\alpha_2) \alpha_4
- (\alpha_2 + \alpha_3 - \alpha_5 + \alpha_6)^2 d_1^2 \geq 0, \\
& 4(\alpha_3 -\alpha_2)[(\alpha_5+\alpha_6)d_1^2 + 2 \alpha_4] - (\alpha_2 + \alpha_3 - \alpha_5 + \alpha_6)^2 d_1^2 \geq 0,
\end{aligned}
\end{equation}
where $d_1 \equiv |d| \geq 0$ was the norm of director vector. Above inequalities held for any $d_1 \in [d_{min}, d_{max}]$, where $[d_{min}, d_{max}]$ stood for the length variation range of different liquid crystal molecules when the temperature was changed.
\end{thm}

We mentioned that Ericksen and Leslie imposed restrictions for above coefficients when the nematic liquid crystals were assumed to be isothermal and incompressible.
We considered a compressible and non-isothermal system, and thus abandoned the assumptions that $\nabla \cdot v=0$ and $|d|=1$.
Please see Lemma \ref{proof of the 2nd law} in Appendix for a proof.

\subsection{Non-isothermal and compressible vectorial models}
In order to derive a concrete model, we presented some classical choices of the entropy function and dissipation matrix as an illustration.
We selected the specific entropy function as
\begin{equation}\label{quad specific entropy}
s(\nu, u, d, \nabla d, C, K, l,h)
=s_0(\nu, u, d, \nabla d)- \frac{1}{2 \beta_1}|C|^2 - \frac{1}{2\beta_2}|K|^2 - \frac{1}{2\beta_3}|l|^2 - \frac{1}{2 \beta_4}|h|^2,
\end{equation}
where the equilibrium entropy $s_0(\nu, u, d, \nabla d)$ was a strictly concave function, $\nu=1/\rho$, and $\beta_1, \cdots, \beta_4$ were positive coefficients related to different relaxation times.
Since the entropy function depended on non-equilibrium variables in a quadratic form, direct calculations showed that
\begin{equation}
 C=-\beta_1 \sigma_V, \quad   K=-\beta_2 \pi_V, \quad  l=- \beta_3 g_V, \quad  h=- \beta_4 q.
\end{equation}
With a diagonal and constant matrix
${\Upsilon} = \frac{1}{\theta}diag(\frac{1}{\gamma_1}, \frac{1}{\gamma_2}, \frac{1}{\gamma_3}, \frac{1}{\theta \gamma_4})$, we had
\begin{equation} \label{Gamma}
{\Upsilon} \cdot
     \begin{pmatrix}
      \sigma_V \\ \pi_V \\g_V\\ q
     \end{pmatrix}
     =
     \begin{pmatrix}
\frac{\sigma_V}{\theta \gamma_1} \\ \frac{\pi_V}{\theta \gamma_2} \\ \frac{g_V}{\theta  \gamma_3} \\ \frac{q}{{\theta}^2  \gamma_4}
\end{pmatrix},
\end{equation}
where the parameters $\gamma_i~(i=1,2,3,4)$ were all positive. $\gamma_1$ and $\gamma_2$ denoted the generalized viscosity, $\gamma_3$ was the rotational friction coefficient, and $\gamma_4$ was the thermal conductivity of liquid crystals.

Now it was time to summarize our new vectorial model for non-isothermal and compressible flows of nematic liquid crystals. This model included the director inertial term and discarded the assumption on $|d|=1$, and thus was applicable to molecular mixtures.
With the above choices of entropy and dissipation matrix, the model read
\begin{equation}\label{nonisothermal EL}
\left\{
\begin{aligned}
&\frac{\partial}{\partial t}\rho + \nabla \cdot (\rho v)=0,\\
&\frac{\partial}{\partial t}(\rho v)+ \nabla\cdot(\rho v \otimes v)
                  =\xi + \nabla \cdot ( -pI+ \sigma_V + \sigma_E+ \sigma_L),\\
&\frac{\partial}{\partial t}(\rho_1 w)+ \nabla\cdot(\rho_1 v \otimes w)= (g_V+g_E+g_L) + \nabla \cdot (\pi_V+\pi_E+\pi_L), \\
&
\frac{\partial}{\partial t}(\rho u)+ \nabla\cdot(\rho v u)=( -pI+ \sigma_V + \sigma_E+ \sigma_L)^T: A + (\pi_V+\pi_E+\pi_L)^T : M \\
&~~~~~~~~~~~~~~~~~~~~~~~~~~~~~ - (g_V+g_E+g_L) \cdot N - \nabla \cdot q,\\
&  \beta_1 [(\rho \sigma_V)_t+\nabla\cdot(\rho v \otimes \sigma_V)] - \theta^{-1}A= - \frac{\sigma_V}{\theta \gamma_1},\\
&       \beta_2 [(\rho \pi_V)_t+\nabla\cdot(\rho v \otimes \pi_V)] - \theta^{-1}M = - \frac{\pi_V}{\theta \gamma_2}, \\
&       \beta_3 [(\rho g_V)_t+\nabla\cdot(\rho v \otimes g_V)] + \theta^{-1} N= - \frac{g_V}{\theta \gamma_3},\\
&       \beta_4 [  (\rho q)_t + \nabla \cdot(\rho v \otimes q)] -\nabla \theta^{-1} = - \frac{q}{{\theta}^2 \gamma_4},\\
& \sigma_E=  \rho \theta (\nabla d \cdot \frac{\partial s}{\partial \nabla d})^T,\quad
\pi_E   = -\rho \theta (\frac{\partial s}{\partial \nabla d})^T,\quad
g_E     =  \rho \theta \frac{\partial s}{\partial d}, \\
&  \sigma_L =\alpha_1 (d^T \cdot A \cdot d) d \otimes d  + \alpha_2 N \otimes d  +\alpha_3 d \otimes N
+ \alpha_4 A  \\
&~~~~~~ +\alpha_5 d \otimes  (A \cdot d)  +\alpha_6 (A \cdot d) \otimes d, \\
&\pi_L=0,\quad
g_L=(\alpha_2-\alpha_3)N + (\alpha_5- \alpha_6)A \cdot d.
\end{aligned}
\right.
\end{equation}
Here $2A=(\nabla v)^T+\nabla v$, $2 \Omega=(\nabla v)^T-\nabla v$, $N=w- \Omega\cdot d$, $M=\nabla w  + \nabla d \cdot \Omega$.
$\theta=(\frac{\partial s}{\partial u})^{-1}$ was the non-equilibrium temperature, and $p=\theta (\frac{\partial s}{\partial {\nu}})$ was the thermodynamical pressure.
The Leslie coefficients $\alpha_1, \cdots, \alpha_6$, as functions of density $\rho$ and non-equilibrium temperature $\theta$, should also satisfy constraints given in \eqref{ineq of 2nd law}.

\subsection{Isothermal and incompressible vectorial models}
In the previous section, we have arrived at a general model for nematic liquid crystals under non-isothermal and compressible conditions.
When liquid crystals were maintained at constant temperature $T$, and the director $d$ was approximated to be of unit length, we arrived at a much simpler model, which has been widely used in previous studies.
For this case, $d(x,t) \in \mathbb {S}^2$, where $\mathbb {S}^2$ denoted a unit sphere of $\mathbb{R}^3$.

Notice that the first law of thermodynamics held automatically under the isothermal condition.
As a consequence, the governing equations reduced to the first three conservation laws in \eqref{nonisothermal EL}.
To show that the classical E-L theory was a special limit of our model, we set the density $\rho \equiv 1$ and further assumed the inertia moment density was negligible under the isothermal and incompressible condition by setting $\rho_1 \rightarrow 0$.
The conservation laws of mass, momentum and angular momentum for hydrodynamics of nematic liquid crystals became
\begin{align}
&\nabla \cdot v =0,\\
&\frac{\partial v}{\partial t}+ v \cdot \nabla v
                  =\xi -\nabla p + \nabla \cdot (\sigma_V + \sigma_E+ \sigma_L),\\
&(g_V+g_E+g_L) + \nabla \cdot (\pi_V+\pi_E+\pi_L)=0 .
\end{align}

Under the isothermal condition, the elastic energy (or free energy) function for nematic liquid crystals was first studied by Oseen and Zocher dated back to the 1920s \cite{Stewart2004The}. Later, Frank \cite{Frank1958I} constructed a concrete energy form based on all possible distortions of the director in 1958.
Following the classical Oseen-Frank elastic energy, we had the equilibrium entropy $s_0$ in \eqref{quad specific entropy} expressed as
\begin{equation}
s_0(\nu, u,d, \nabla d)|_{\nu=1, \theta=T}=-\frac{1}{T} W_F(d, \nabla d),
\end{equation}
in which
\begin{equation}\label{Oseen-Frank}
2 W_F(d, \nabla d)=k_1 (\nabla \cdot d)^2 +k_2 [d \cdot (\nabla \times d)]^2 + k_3|d \times (\nabla \times d)|^2+
(k_2+k_4)[tr(\nabla d)^2 - (\nabla \cdot d)^2] .
\end{equation}
Here $k_1, k_2, k_3, (k_2+k_4)$, depending on the density $\rho$ and temperature $T$, stood for the pure splay, pure twist, pure bend and saddle-splay constants respectively.
Furthermore, since the integral of the fourth term of Eq. \eqref{Oseen-Frank}
\begin{equation*}
(k_2+k_4)\int [tr(\nabla d)^2 - (\nabla \cdot d)^2] dx,
\end{equation*}
depended solely on boundary values of $d$, the saddle-splay term could be discarded for planar boundary conditions \cite{Ball2017Mathematics}.
We suggested refs \cite{Durand1969Quasielastic, Lee1986Computations,Allen1988Calculation} for experimental, theoretical and numerical investigations on the elastic constants $k_i$ $(i=1,2,3,4)$.

In this case, we had $\frac{\partial s}{\partial d}=-\frac{1}{T}\frac{\partial W_F}{\partial d},~
\frac{\partial s}{\partial \nabla d}=-\frac{1}{T}\frac{\partial W_F}{\partial \nabla d},$
then the Ericksen parts of the stress tensor, the director surface torque and the director body torque in Eq. \eqref{ericksen} became
\begin{equation}\label{isothermal ericksen}
\sigma_E= -  (\nabla d \cdot \frac{\partial W_F}{\partial \nabla d})^T,\quad
\pi_E   =    (\frac{\partial W_F}{\partial \nabla d})^T,\quad
g_E     = -  \frac{\partial W_F}{\partial d},
\end{equation}
by noticing the fact that the non-equilibrium temperature coincided with the equilibrium temperature under isothermal conditions, $\theta=T$.
In the stationary limit, $\beta_1, \beta_2, \beta_3 \rightarrow 0$,
by choosing the dissipation matrix as in Eq. \eqref{Gamma} and taking Maxwellian iteration, we had
\begin{equation}
\sigma_V=\gamma_1 A, \quad \pi_V= \gamma_2 M, \quad g_V= -\gamma_3 N.
\end{equation}

Summarizing above equations, we arrived at the classical E-L theory for nematic liquid crystals under the isothermal and incompressible condition as
\begin{equation} \label{isothermal EL}
\left\{
\begin{aligned}
&\nabla \cdot v =0,\\
&\frac{\partial v}{\partial t}+ v \cdot \nabla v
                  =\xi -\nabla p + \nabla \cdot (\sigma_V + \sigma_E+ \sigma_L), \\
&\sigma_V+\sigma_E+ \sigma_L
= -(\nabla d \cdot \frac{\partial W_F}{\partial \nabla d})^T +\alpha_1 (d^T \cdot A \cdot d) d \otimes d  + \alpha_2 N \otimes d  +\alpha_3 d \otimes N
\\
&~~~~~~~~~~~~~~~~~~~~~~+ (\alpha_4 + \gamma_1) A+\alpha_5 d \otimes  (A \cdot d)  +\alpha_6 (A \cdot d) \otimes d,\\
&-\frac{\partial W_F}{\partial d} + \nabla \cdot (\frac{\partial W_F}{\partial \nabla d})^T
   -(\alpha_3 -\alpha_2 + \gamma_3) N - (\alpha_6 -\alpha_5) A\cdot d + \gamma_2 \nabla \cdot M =0
,
\end{aligned}
\right.
\end{equation}
except for an additional term $(\gamma_2 \nabla \cdot M)$ in angular momentum equation and
two extra terms
as $\gamma_1 A$ and $-\gamma_3  N$, which were absorbed into the classical coefficients as $\alpha_4 + \gamma_1$ and $\alpha_3 -\alpha_2 + \gamma_3$.
The first two terms of angular momentum equation were also known as the molecular field of liquid crystals, whose components read
$-\frac{\delta W_F}{\delta d_i}
=-\frac{\partial W_F}{\partial d_i} + \partial_j (\frac{\partial W_F}{\partial d_{i,j}})$.
The additional term $\partial_j M_{ji} = \partial_j (w_{i,j}+ \Omega_{ki} d_{k,j})$ described how the spatial gradients of director velocity affected the balance of angular momentum, which was not taken into consideration in classical E-L equations of liquid crystals \cite{Chandrasekhar1992Liquid,Stewart2004The}.

\section{Extension to the tensorial theory} \label{Q-tensor}
In the previous section, a generalized vectorial model for non-isothermal flows of nematic liquid crystals was derived, which included the classical E-L model as a special case under the isothermal condition.
In this part, we focused on the Conservation-Dissipation Formalism for the generalized Qian-Sheng (Q-S) model \cite{Qian1998Generalized}, and the Beris-Edwards model \cite{Edwards1990Generalized, Beris1994Thermodynamics} could be treated in a similar way.
From the multi-scale thermodynamic \cite{Grmela2018Generic} point of view, the E-L like vectorial theory belonged to a description of lower level, while the tensorial theory contained more details of the internal structure of liquid crystals, and thus belonged to a higher level.
The CDF provided a unified and elegant framework for describing non-isothermal systems for both the vectorial model and the tensorial model.
To avoid lengthy derivations, here only the key results for the tensorial model were included.
Interested readers might turn to the SI for details.

In the tensorial theory, the conservation laws of mass and momentum for non-isothermal compressible flows of nematic liquid crystals were the same as that for the vectorial model, while the angular momentum and total energy conservation laws were different, which read
\begin{align}
&\frac{\partial}{\partial t}(\rho_1 \dot Q)+ \nabla \cdot(\rho_1 v \otimes \dot Q) = g + \nabla \cdot \pi, \label{Qangmom} \\
&\frac{\partial}{\partial t}(\rho e) + \nabla \cdot(\rho v  e)=\xi \cdot v + \nabla \cdot (\sigma \cdot v + \pi :  \dot Q - q)\label{Qenergy}.
\end{align}
where the second-order tensor $g$ was body torque and the third-order tensor $\pi$ was surface torque.
The total energy included internal energy, translational kinetic energy and rotational kinetic energy as $\rho e = \rho u + \frac{1}{2}\rho v^2 + \frac{1}{2}\rho_1 {\dot Q}^2$.
In analogy with Eq. \eqref{sigma_EL} in the vectorial theory, the stress tensor, $Q$-tensor body torque and surface torque could be separated into
$\sigma= -pI + \sigma_E + \sigma_{QS} + \sigma_{V}$, $g=g_E+g_{QS}+g_V$ and $\pi=\pi_E+\pi_{QS}+\pi_V$, respectively.

Supposing the balance equation of internal energy to satisfy the principle of objectivity,
direct calculations yielded
\begin{align}\label{Qenergy2}
\rho \frac{du}{dt}&=\sigma^T: {\nabla v} - g : \dot Q  + \pi^T \vdots \nabla \dot Q - \nabla \cdot q
=\sigma^T:A - g: Y + \pi^T \vdots M - \nabla \cdot q,
\end{align}
where $\pi^T \vdots \nabla \dot Q=(\pi^T)_{ijk} \partial_k (\dot Q_{ji})$ was the triple scalar product between third-order tensors $(\pi^T)_{ijk}$ and $\partial_k \dot Q_{ji}=\partial \dot Q_{ji} / \partial x_k$.
The second-order tensor $Y=\dot Q - \Omega \cdot Q + Q \cdot \Omega$ was the Jaumann derivative of $Q$, representing the change rate of $Q$ relative to the fluid angular velocity.
In Cartesian coordinates, the third-order tensor $M$ was defined as
\begin{equation}
M_{kij}= \partial_k \dot Q_{ij}  + \partial_k Q_{il} \Omega_{lj} - \Omega_{il} \partial_k Q_{lj},
\end{equation}
which was shown to be objective in Lemma \ref{tensor M anti}.
The terms violated the principle of frame-indifference in Eq. \eqref{Qenergy2} were neglected.
The third-order tensors $\pi$ and $M$ in Eqs. \eqref{Qangmom}-\eqref{Qenergy2}, accounting for the contributions of surface torque to angular momentum and energy, were not found in existing references of tensorial models as far as we know.

In order to find constitutive relations for $(\sigma, \pi, g, q)$, we followed the CDF and specified a strictly concave mathematical entropy function
\begin{equation}\label{entropy_Q}
\eta=\rho s(\nu, u, Q, \nabla Q, \tilde C, \tilde K, \tilde l, \tilde h),
\end{equation}
where $\nu=1/\rho$.
$(\nu, u, Q, \nabla Q)$ denoted the conserved variables, while $(\tilde C, \tilde K, \tilde l, \tilde h)$ were specified such that $\frac{\partial s}{\partial \tilde C}=\sigma_V, \frac{\partial s}{\partial \tilde l}=g_V, \frac{\partial s}{\partial \tilde K}=\pi_V$ and $\frac{\partial s}{\partial \tilde h}=q$.
According to the generalized Gibbs relation and Lemma \ref{tensor_Q_symm}, we derived the time evolution equation of the entropy $\eta$ as
\begin{align*}
& \mathcal {D} s(\nu, u, Q, \nabla Q,\tilde C, \tilde K, \tilde l, \tilde h)=\nabla \cdot J^f + \Sigma^f,
\end{align*}
where $J^f=- \theta ^{-1} q$ and
\begin{align*}
\Sigma^f&= \underline{ (\theta^{-1} \sigma_E^T - \rho \nabla Q : \frac{\partial s}{\partial   \nabla Q} ) : A
- ( \theta^{-1} g_E -  \rho \frac{\partial s}{\partial Q}) : Y  + [ \theta^{-1} \pi^T_E + \rho (\frac{\partial s}{\partial \nabla Q}) ] \vdots M }\\
 & +\underline{ (\frac{\partial s}{\partial \tilde C})^T: \mathcal D \tilde C + \theta^{-1} \sigma_V^T :A }
   +\underline{ (\frac{\partial s}{\partial \tilde l})^T: \mathcal D \tilde l - \theta^{-1} g_V :Y}
   +\underline{ (\frac{\partial s}{\partial \tilde K})^T \vdots \mathcal D \tilde K  + \theta^{-1} \pi^T_V \vdots M }\\
 &
   +\underline{ \frac{\partial s}{\partial \tilde h} \cdot  \mathcal D \tilde h  +  q \cdot \nabla {\theta ^{-1}} }
   +\underline{ \theta^{-1} (\sigma_{QS}^T: A - g_{QS} : Y + \pi^T_{QS} \vdots M) }.
\end{align*}

Following basically the same procedure taken in the vectorial theory, we arrived at a new system of constitutive equations for the non-isothermal tensorial model as
\begin{equation}\label{nonisothermal tensorial model}
\left\{
\begin{aligned}
&\sigma_E =  \rho \theta (\nabla Q : \frac{\partial s}{\partial   \nabla Q})^T , \quad
g_E      = \rho \theta \frac{\partial s}{\partial Q}, \quad
\pi_E = - \rho \theta (\frac{\partial s}{\partial \nabla Q})^T,
\\
&  [({\rho \tilde C})_t  + \nabla \cdot (\rho v \otimes \tilde C)] + \theta^{-1} A = \frac{\sigma_V}{\theta \bar \gamma_1},  \\
&  [({\rho \tilde l})_t       + \nabla \cdot (\rho v \otimes \tilde l)]   - \theta^{-1} Y = \frac{g_V}{\theta \bar \gamma_2},  \\
&  [({\rho \tilde K})_t  + \nabla \cdot (\rho v \otimes \tilde K)] + \theta^{-1} M = \frac{\pi_V}{\theta \bar \gamma_3},  \\
&  [({\rho \tilde h})_t  + \nabla \cdot (\rho v \otimes \tilde h)] + \nabla \theta^{-1} = \frac{q}{\theta^2 \bar \gamma_4}, \\
&\sigma_{QS}=\bar \alpha_1 (Q: A) Q +  \bar \alpha_4 A + \bar \alpha_5 Q\cdot A + \bar \alpha_6 A \cdot Q + \frac{1}{2} \mu_2 Y - \mu_1 Q \cdot Y + \mu_1 (Q \cdot Y)^T \\
 &+ \bar \alpha_7 tr(A) I + \bar \alpha_8 tr(A) Q \cdot Q , \\
&g_{QS}=-\frac{1}{2} \mu_2 A - \mu_1 Y, \quad \pi_{QS}=0,
\end{aligned}
\right.
\end{equation}
where $\bar \alpha_1$, $\bar \alpha_4$, $\bar \alpha_5$, $\bar \alpha_6$, $\mu_1$, $\mu_2$, $\bar \alpha_7$ and $\bar \alpha_8$ were material and temperature-dependent viscosity coefficients of liquid crystals, and $\bar \alpha_6 -\bar \alpha_5=\mu_2$. We adopted $\bar \alpha_7$ and $\bar \alpha_8$ to describe the compressible effects on the stress tensor, which vanished if the system was incompressible.
In that case, the expression of $\sigma_{QS}$ reduced to the classical constitutive equations in ref. \cite{Qian1998Generalized}.

Now, the total entropy production rate of system \eqref{nonisothermal tensorial model} was readily obtained as
\begin{align*}
\Sigma^f
=&\frac{1}{\theta \bar {\gamma_1}} |\sigma_V|^2 +  \frac{1}{\theta \bar {\gamma_2}} |g_V|^2 +\frac{1}{\theta \bar {\gamma_3}} |\pi_V|^2 +\frac{1}{\theta^2 \bar {\gamma_1}} |q|^2
+\underline
{ \bar \alpha_7 [tr(A)]^2 + \bar \alpha_8 tr(A)Q : (Q \cdot A)
} \\
&+\underline{
\bar \alpha_1 (Q \cdot A)^2 + \bar \alpha_4 |A|^2 + (\bar \alpha_5 +\bar \alpha_6)(Q \cdot A):A + \mu_2 Y:A + \mu_1 |Y|^2
},
\end{align*}
where the first four terms $(epr1+ epr2 + epr3 + epr4) \geq 0$ were produced by relaxation processes.
The underlined residual terms, denoted as $epr5$, corresponded to the classical Qian-Sheng theory.
In order to guarantee the non-negativeness of $epr5$, the phenomenological coefficients $\bar \alpha's$ and $\mu's$ should satisfy further restrictions,
similar as in the Theorem \ref{2nd law} of vectorial model.
Considering the complexity of tensor calculations, we presented a non-trivial sufficient condition for $epr5 \geq 0$ as
\begin{equation} \label{epr5 geq 0}
\begin{split}
&
\bar \alpha_8=0, \quad \bar \alpha_7 \geq 0,
\quad \bar \alpha_4 \geq 0,
\quad \mu_1 \geq 0,
\quad \bar \alpha_1 \geq 0, \\
&
\quad 4 \bar \alpha_4 \mu_1 - {\mu_2^2} \geq 0,
\quad 4 \bar \alpha_4 \bar \alpha_1   - ({\bar \alpha_5 +\bar \alpha_6})^2 \geq 0, \\
&\quad 4 \bar \alpha_1 \bar \alpha_4 \mu_1 - \bar \alpha_1 \mu_2^2 - \mu_1({\bar \alpha_5 +\bar \alpha_6})^2 \geq 0.
\end{split}
\end{equation}
Please see Lemma \ref{proof of the epr5 geq 0} for a proof of the sufficiency.

Especially, under the isothermal condition, the temperature became constant $\theta=T$.
Further setting the density as $\rho=1$, then the classical Landau-de Gennes energy density was expressed as \cite{deGennes1995The}
\begin{align}
F_{LdG}(Q, \nabla Q)=& -\frac{a}{2} tr(Q^2) -\frac{b}{3} tr(Q^3)  + \frac{c}{4}tr(Q^4)  \nonumber \\
&+ \frac{1}{2}(L_1 |\nabla Q|^2 + L_2 Q_{ij,j}Q_{ik,k} +L_3 Q_{ij,k}Q_{ik,j} +L_4 Q_{ij}Q_{kl,i}Q_{kl,j}),
\end{align}
where the bulk energy included the phenomenological coefficients $a,b,c>0$, and the elastic energy $L_i$ ($i=1,2,3,4$).
By choosing the entropy function as $s=-F_{LdG}(Q, \nabla Q)/T$, the balance equations of momentum and angular momentum reduced to
\begin{align}
\rho \frac{dv}{dt}
=&\xi - \nabla p - \nabla \cdot (\nabla Q : \frac{\partial F_{LdG}}{\partial   \nabla Q})^T + \bar \alpha_1 (Q: A) Q +  (\bar \alpha_4 + \bar \gamma_1) A +  \bar \alpha_5 Q\cdot A \\
& + \bar \alpha_6 A \cdot Q + \frac{1}{2} \mu_2 Y - \mu_1 Q \cdot Y + \mu_1 (Q \cdot Y)^T, \\
\rho_1 \ddot Q
=&-\frac{\partial F_{LdG}}{\partial Q} + \nabla \cdot (\frac{\partial F_{LdG}}{\partial \nabla Q} )^T -\frac{1}{2} \mu_2 A - (\mu_1-\bar \gamma_2) Y  + \bar \gamma_3 \nabla \cdot M.
\end{align}
The above equations were essentially the same as the Qian-Sheng model \cite{Qian1998Generalized}, except that we had two extra terms $\bar \gamma_1 A, -\gamma_2 Y$, and included $\nabla \cdot M$.
The last term $\partial_k M_{kij}= \partial_k( \partial_k \dot Q_{ij} - \Omega_{il} Q_{lj,k} + Q_{il,k} \Omega_{lj})
$ was frame-indifference, and described the effect of spatial anisotropy of tensorial order parameter on the balance of angular momentum.

\begin{rem}
Based on the analysis of length and time scales \cite{Qian1998Generalized}, the E-L equations reflected the long-range and slow-motion limit of the $Q$-tensor model.
By assuming a uniaxial symmetry of $Q$-tensor,
\begin{equation}
Q_{ij}(x,t)=\frac{S_0}{2} [3d_i(x,t) d_j(x,t) - \delta_{ij}],
\end{equation}
it was shown that the classical isothermal E-L equations could be recovered from the Qian-Sheng model \cite{Qian1998Generalized, Han2015From}.
A similar conclusion was expected for our generalized non-isothermal tensorial and vectorial models, since both of them could be casted into the framework of CDF.
\end{rem}

\section{Conclusion}
In the first paper of this series \cite{Peng2018Conservation}, isothermal models for particle diffusion in dilute solutions, polymer phase separation dynamics and simplified Lin-Liu model for isothermal flows of nematic liquid crystals were studied under the guidance of Conservation-Dissipation Formalism. In the current paper, more complicated non-isothermal and compressible situations were examined by taking nematic liquid crystals as a typical example.

Starting from the fundamental conservation laws, a generalized vectorial model for flows of nematic liquid crystal in non-isothermal environment was derived within the framework of CDF. The hydrodynamic equations thus deduced were fully consistent with the first and second laws of thermodynamics. A necessary and sufficient condition involving inequalities of Leslie coefficients was obtained to guarantee the non-negativeness of entropy production rate, with emphasis on different molecular lengths. Under the isothermal and incompressible condition, the classical E-L theory was shown to be a special case of our new vectorial model in the stationary limit. Moreover, all results were readily extended from a vectorial form to the tensorial theory under non-isothermal and compressible conditions by making use of CDF.

It was widely recognized that the non-isothermal situations would bring some intrinsic difficulties to the thermodynamics-based mathematical modeling, like proper definitions of the non-equilibrium temperature, the non-violation of first law of thermodynamics, the non-Fourier's law for heat conduction, etc.
Many non-equilibrium thermodynamic theories or approaches, like the CIT \cite{de2013non}, the Doi's variational approach \cite{doi2013soft}, the energetic variational framework \cite{Sun2008On} and so on, might fail to overcome these obstacles.
In contrast, as one of its outstanding merits, CDF was shown to be able to handle various non-isothermal irreversible processes easily in a systematic way.
The first and second laws of thermodynamics were respected in CDF simultaneously.
The non-equilibrium temperature and pressure were introduced in a self-consistent way as the partial derivatives of the entropy.
Relaxation-type equations for the stress tensor and heat conduction were constructed as direct generalizations of Newton's law for viscosity and Fourier's law for heat conduction.
All these facts demonstrated that CDF would have a broader scope of applications in soft matter physics.

\section*{acknowledgment}
This work was supported by the 13th 5-Year Basic Research Program of CNPC (2018A-3306), the National Natural Science
Foundation of China (Grants 21877070) and Tsinghua University Initiative Scientific Research Program (Grants 20151080424).

\section{Appendix A: Some Lemmas}

\begin{lemma} \label{lemma of Dd}
(A Proof of Eq. \eqref{Dd})
The director vector $d$ and its spatial gradient $\nabla d$ satisfied the following differential relations:
\begin{equation*}
\mathcal D d= \rho(N+\Omega \cdot d), \quad
\mathcal D (\nabla d) = \rho[ M -(\nabla d) \cdot \Omega + \Omega \cdot(\nabla d)-A \cdot (\nabla d)].
\end{equation*}
\end{lemma}

\begin{proof}
Firstly, we had
\begin{align*}
\mathcal D d
&=\frac{\partial }{\partial t}(\rho d) + \nabla \cdot (v \rho d )
=\rho ( \frac{\partial}{\partial t}d + v \cdot \nabla d)
=\rho \frac{d}{dt}d =\rho w\\
&
=\rho(N+\Omega \cdot d),
\end{align*}
by recalling the definition of  relative angular velocity $N=w - \Omega \cdot d$.
Secondly, we deduced that
\begin{align*}
\mathcal D (\nabla d)
=&\rho (\frac{\partial }{\partial t} + v \cdot \nabla )(\nabla d)
= \rho \nabla[ (\frac{\partial }{\partial t} + v \cdot \nabla) d] - \rho \nabla v \cdot \nabla d  \\
=& \rho \nabla w - \rho \nabla v \cdot \nabla d
= \rho[ M -(\nabla d) \cdot \Omega + \Omega \cdot(\nabla d)-A \cdot (\nabla d)],
\end{align*}
where the definition $M=\nabla w + (\nabla d) \cdot \Omega$ was used in the last step.
\end{proof}

The following identity was obtained by Ericksen \cite{Ericksen1961Conservation} in 1961, we derived it here to preserve the integrity of contents.
\begin{lemma}\label{tensor symmetric}
If $s=s(\nu, u, d, \nabla d, C, K, l,h)$ was objective with respect to conserved variables,
the tensor $(\frac{\partial s}{\partial d}\otimes d +  \frac{\partial s}{\partial \nabla d} \cdot \nabla d - \nabla d \cdot  \frac{\partial s}{\partial \nabla d}  )$ must be symmetric \cite{Ericksen1961Conservation}.
\end{lemma}
\begin{proof}
We considered Euclidean transformation
\begin{equation}\label{rigid rotation}
x^*(t)= \mathcal R(t) x(t) + c(t),
\end{equation}
of a continuum,
here $\mathcal R=(\mathcal{R}_{ij})$ stood for a second-order orthogonal tensor function, $i.e.$, $\mathcal R^T \mathcal R=I$.
States after rigid motions and original states were denoted by asterisked and unasterisked variables respectively.
Mathematically, we had
\begin{equation*}
s(\nu, u, d^*, (\nabla d)^*, C, K, l,h)=s(\nu, u, d, \nabla d, C, K, l,h),
\end{equation*}
where
\begin{equation}\label{rigid}
d^*=\mathcal R \cdot d, \quad (\nabla d)^*=\mathcal R \cdot (\nabla d) \cdot {\mathcal R}^T.
\end{equation}
Specifically, we chose arbitrary infinitesimal rotations with
\begin{equation}
\mathcal R_{ij}= \delta_{ij}+\varepsilon_{ij}, \quad
\varepsilon_{ji}=-\varepsilon_{ij}, \quad
\varepsilon_{ij}\varepsilon_{kl} \approx 0.
\end{equation}
Substituting above relations into Eq. \eqref{rigid}, we had
\begin{equation*}
d^*-d=\varepsilon \cdot d, \quad
(\nabla d)^* -\nabla d
=\varepsilon \cdot \nabla d - \nabla d \cdot \varepsilon - \varepsilon \cdot \nabla d \cdot \varepsilon
= \varepsilon \cdot \nabla d - \nabla d \cdot \varepsilon + \mathcal{O} (\varepsilon^2),
\end{equation*}
where the second equality held within the leading order of $\varepsilon$.
Therefore, one deduced approximately that
\begin{equation*}
0\approx
\frac{\partial s}{\partial d} \cdot (d^* -d)
+ (\frac{\partial s}{\partial \nabla d})^T : [(\nabla d)^* - \nabla d]
= - (\frac{\partial s}{\partial d}\otimes d +  \frac{\partial s}{\partial \nabla d} \cdot \nabla d - \nabla d \cdot  \frac{\partial s}{\partial \nabla d}  ): \varepsilon .
\end{equation*}
Since the above equality held for arbitrary skew-symmetric tensors $\varepsilon$, the tensor in the parenthesis must be symmetric.
This completed the proof.
\end{proof}

\begin{lemma}\label{proof of the 2nd law}
(A Proof of the Theorem \ref{2nd law})
\end{lemma}
\begin{proof}
Choose a set of proper orthogonal basis $(e_1, e_2, e_3)$, such that
\begin{equation*}
d=d_1  e_1, \quad N= N_1 e_1 + N_2 e_2, \quad A=A_{ij}e_i \otimes e_j,
\end{equation*}
where $d_1=|d| \geq 0$ was the norm of director vector $d(x,t)$,
and the repeated indices were subject to the summation convention.
Moreover, two vectors $d$ and $N$ were in a subspace spanned by $(e_1, e_2)$.
Inserting the relations of $d$, $N$ and $A$ into $epr5$, we had
\begin{align*}
&\theta \cdot epr5\\
=&\alpha_{1} d_1^2 A_{11}^2+(\alpha_2 + \alpha_3 - \alpha_5 + \alpha_6)d_1(A_{11} N_1 + A_{12}N_2)
+\alpha_4 A_{ij}A_{ij} \\
&+(\alpha_5+\alpha_6)d_1^2(A_{11}^2+A_{12}^2+A_{13}^2) + (\alpha_3-\alpha_2)(N_1^2+N_2^2)\\
=&\alpha_{1} d_1^2 A_{11}^2+(\alpha_2 + \alpha_3 - \alpha_5 + \alpha_6)d_1(A_{11} N_1 + A_{12}N_2)
+\alpha_4 (A_{11}^2+A_{22}^2+A_{33}^2) +2\alpha_4 (A_{12}^2+ A_{13}^2+A_{23}^2) \\
&+(\alpha_5+\alpha_6)d_1^2(A_{11}^2+A_{12}^2+A_{13}^2) + (\alpha_3-\alpha_2)(N_1^2+N_2^2)\\
=&[(\alpha_1+\alpha_5+\alpha_6)d_1^2 +\alpha_4]A_{11}^2
+(\alpha_2 + \alpha_3 - \alpha_5 + \alpha_6) d_1 A_{11} N_1 + (\alpha_3 - \alpha_2) N_1^2\\
&+[(\alpha_5 + \alpha_6)d_1^2 +2 \alpha_4]A_{12}^2
+(\alpha_2 + \alpha_3 - \alpha_5 + \alpha_6)d_1 A_{12} N_2 + (\alpha_3 - \alpha_2) N_2^2\\
&+(\alpha_5  + \alpha_6)A_{13}^2
+ \alpha_4 A_{22}^2
+ 2 \alpha_4 A_{23}^2
+ \alpha_4 A_{33}^2 \geq 0,
\end{align*}
where we used the symmetric property of $A$ in the second step.
Since $A_{ij}(i\leq j), N_1, N_2$ and $d_1$ were all independent variables, it required that
\begin{align*}
&[(\alpha_1+\alpha_5+\alpha_6)d_1^2 +\alpha_4]A_{11}^2
+(\alpha_2 + \alpha_3 - \alpha_5 + \alpha_6) d_1 A_{11} N_1 + (\alpha_3 - \alpha_2) N_1^2\\
&+[(\alpha_5 + \alpha_6)d_1^2 +2 \alpha_4]A_{12}^2
+(\alpha_2 + \alpha_3 - \alpha_5 + \alpha_6)d_1 A_{12} N_2 + (\alpha_3 - \alpha_2) N_2^2 \geq 0,\\
& (\alpha_5  + \alpha_6)A_{13}^2 \geq 0, \quad  \alpha_4 A_{22}^2 \geq 0, \quad  2 \alpha_4 A_{23}^2 \geq 0, \quad  \alpha_4 A_{33}^2 \geq 0 .
\end{align*}
The last four inequalities yielded that
\begin{equation} \label{ineq of alpha_4}
\alpha_5  + \alpha_6 \geq 0, \quad \alpha_4 \geq 0.
\end{equation}
Moreover, the left-hand side of the first inequality corresponded to a quadratic form of variables $(A_{11}, A_{12}, N_1, N_2)$, where the associated symmetric matrix read
\begin{align*}
\mathcal {M}
=
\begin{pmatrix}
(\alpha_1+\alpha_5+\alpha_6)d_1^2 +\alpha_4 & 0 & \frac{(\alpha_2 + \alpha_3 - \alpha_5 + \alpha_6)d_1}{2}  & 0
\\  & (\alpha_5 + \alpha_6)d_1^2 +2 \alpha_4 & 0 & \frac{(\alpha_2 + \alpha_3 - \alpha_5 + \alpha_6) d_1}{2} \\
* & & \alpha_3 -\alpha_2 & 0\\
& & & \alpha_3 -\alpha_2
\end{pmatrix}.
\end{align*}
The matrix $\mathcal {M}$ was semi-positive if and only if all leading principal minors had non-negative determinants.
Therefore,
\begin{equation}\label{ineq of alpha_1}
\begin{aligned}
&(\alpha_1+\alpha_5+\alpha_6)d_1^2 +\alpha_4 \geq 0,\\
&(\alpha_5 + \alpha_6)d_1^2 +2 \alpha_4      \geq 0, \\
& 4(\alpha_3 -\alpha_2)(\alpha_1+\alpha_5+\alpha_6)d_1^2 + 4(\alpha_3 -\alpha_2) \alpha_4
- (\alpha_2 + \alpha_3 - \alpha_5 + \alpha_6)^2 d_1^2 \geq 0, \\
& 4(\alpha_3 -\alpha_2)[(\alpha_5+\alpha_6)d_1^2 + 2 \alpha_4] - (\alpha_2 + \alpha_3 - \alpha_5 + \alpha_6)^2 d_1^2 \geq 0,
\end{aligned}
\end{equation}
where the above inequalities involving $\alpha_1, \cdots, \alpha_6$ held for any $d_1 \in [d_{min}, d_{max}]$, and $[d_{min}, d_{max}]$ stood for the length variation range of liquid crystal molecules when the temperature was changed.
A combination of relations \eqref{ineq of alpha_4} and \eqref{ineq of alpha_1} yielded that
\begin{equation}
\begin{aligned}
& \alpha_4 \geq 0, \quad \alpha_5 + \alpha_6\geq 0, \\
&(\alpha_1+\alpha_5+\alpha_6)d_1^2 +\alpha_4 \geq 0,\\
& 4(\alpha_3 -\alpha_2)(\alpha_1+\alpha_5+\alpha_6)d_1^2 + 4(\alpha_3 -\alpha_2) \alpha_4
- (\alpha_2 + \alpha_3 - \alpha_5 + \alpha_6)^2 d_1^2 \geq 0, \\
& 4(\alpha_3 -\alpha_2)[(\alpha_5+\alpha_6)d_1^2 + 2 \alpha_4] - (\alpha_2 + \alpha_3 - \alpha_5 + \alpha_6)^2 d_1^2 \geq 0.
\end{aligned}
\end{equation}
\end{proof}

\begin{lemma}\label{tensor M anti}
The third-order tensor $M$, corresponding to $\nabla \dot Q$, defined as
\begin{equation*}
M_{kij}= \partial_k \dot Q_{ij} - \Omega_{il} Q_{lj,k} + Q_{il,k} \Omega_{lj},
\end{equation*}
was objective.
\end{lemma}
\begin{proof}
A third-order tensor was objective if and only if it transformed according to the rules
\begin{equation*}
M_{kij}^*={\mathcal R}_{kr} {\mathcal R}_{ip} {\mathcal R}_{jq}M_{rpq} ,
\end{equation*}
here ${\mathcal R}={\mathcal R}(t)$ was a proper orthogonal time-dependent tensor. Since the Jaumann's derivative of tensor $Q$, defined as $Y=\dot Q - \Omega \cdot Q + Q \cdot \Omega$, was objective, we had $Y_{ij}^*={\mathcal R}_{ip} {\mathcal R}_{jq} Y_{pq}$. Taking spatial derivatives on each sides and recalling the Euclidean transformation $x^*(t)= \mathcal R(t) x(t) + c(t)$, one had
\begin{equation*}
\partial_{k}^* Y_{ij}^*= \frac{\partial }{\partial x_k^*} Y_{ij}^*
= \frac{\partial }{\partial x_r}\frac{\partial x_r}{\partial x_k^*} ({\mathcal R}_{ip} {\mathcal R}_{jq} Y_{pq})
={\mathcal R}_{kr} {\mathcal R}_{ip} {\mathcal R}_{jq} \partial_r Y_{pq} .
\end{equation*}
Substitutions of the definition $Y_{ij}^*= {\dot Q_{ij}}^* - \Omega_{ip}^* Q_{pj}^* + Q_{ip}^* \Omega^*_{pj} $ gave that
\begin{align*}
&\partial_{k}^* {\dot Q_{ij}}^* - \Omega_{ip}^* \partial_{k}^* Q_{pj}^* + \partial_{k}^* Q_{ip}^* \Omega^*_{pj}
- {\mathcal R}_{kr}  {\mathcal R}_{ip} {\mathcal R}_{jq}({\partial_{r} \dot Q_{pq}} - \Omega_{pt} \partial_{r} Q_{tq} + \partial_{r} Q_{pt} \Omega_{tq})\\
=& \partial_{k}^* \Omega_{ip}^* Q_{pj}^* - Q_{ip}^* \partial_{k}^* \Omega^*_{pj}
- {\mathcal R}_{kr} {\mathcal R}_{ip} {\mathcal R}_{jq} (\partial_r \Omega_{pt} Q_{tq}  -  Q_{pt} \partial_r \Omega_{tq}) .
\end{align*}
The left-hand side was just $M_{kij}^* -{\mathcal R}_{kr} {\mathcal R}_{ip} {\mathcal R}_{jq} M_{rpq} $, which meant that the objectivity of $M$ was equivalent with cancelation of the right-hand side. Actually, we could prove $\partial_{k}^* \Omega_{ip}^* Q_{pj}^* =  {\mathcal R}_{kr}  {\mathcal R}_{ip} {\mathcal R}_{jq} (\partial_r \Omega_{pt} Q_{tq})$, and the other relation followed immediately.

By taking the time derivative of Euclidean transformation $x_i^*= \mathcal R_{ij} x_j + c_i$, one had $v_i^* = \mathcal R_{ij} v_j + \dot{\mathcal  R}_{ij} x_j+ \dot c_i$. Further taking spatial gradients, we obtained $\partial_p^* v_i^* = \mathcal R_{pk} \dot{\mathcal  R}_{ik} + \mathcal R_{pk} \partial_k v_j \mathcal R_{ij} $.
It could be readily shown that the anti-symmetric part of velocity gradient was non-objective and transformed like,
\begin{equation*}
\Omega_{ip}^* = {\mathcal R}_{is} \Omega_{sa} {\mathcal R}_{pa} + { \mathcal {\dot R}}_{ik} \mathcal R_{pk},
\end{equation*}
via the relation ${ \mathcal {\dot R}}_{ik} \mathcal R_{pk}=-{ \mathcal {\dot R}}_{pk} \mathcal R_{ik}$.
Taking spatial gradients on both sides and multiplying by $ Q_{pj}^* = {\mathcal R}_{pt} Q_{tq} R_{jq}$ (Q was assumed to be objective) , we had
\begin{equation*}
\Omega_{ip, k}^*  Q_{pj}^*
= {\mathcal R}_{is} {\mathcal R}_{pa} {\mathcal R}_{kb}  {\mathcal R}_{pt} R_{jq} (\Omega_{sa,b} Q_{tq})
= \delta_{at}  {\mathcal R}_{is} {\mathcal R}_{jq} {\mathcal R}_{kb} (\Omega_{sa,b} Q_{tq})
={\mathcal R}_{ip} {\mathcal R}_{jq} {\mathcal R}_{kr} (\Omega_{pt,r} Q_{tq})
,
\end{equation*}
where we used the orthogonality ${\mathcal R}_{pa}{\mathcal R}_{pt} = \delta_{at}$ in the second step and relabeled indices in the last one.
This completed the proof.
\end{proof}

\begin{lemma}\label{tensor_Q_symm}
In the tensorial theory, by assuming the entropy function $s(\nu, u, Q, \nabla Q, \sigma_V, g_V, \pi_V, q)$ was objective with respect to conserved variables,
the tensor
$[2 Q_{ik} \frac{\partial s}{\partial Q_{jk}}
 +2 (\partial_k Q_{il}) \frac{\partial s}{\partial (\partial_k Q_{jl})}
 +( \partial_{i} Q_{kl})\frac{\partial s}{\partial (\partial_j Q_{kl})}]$
must be symmetric, and its contraction with the anti-symmetric part of velocity gradient $\Omega$ vanished.
\end{lemma}
\begin{proof}
As in Lemma \ref{tensor symmetric}, we chose arbitrary infinitesimal rotations with
\begin{equation*}
{\mathcal R}_{ij} = \delta_{ij } + \varepsilon_{ij},  \quad \varepsilon_{ji}= - \varepsilon_{ij}, \quad \varepsilon_{ij}\varepsilon_{kl} \approx 0,
\end{equation*}
and denoted variables under the motion $x^*(t)= {\mathcal R}(t) x(t)+c(t)$ by a starred symbol. Taking series expansion of $\varepsilon$ to the first order, we had
\begin{align*}
Q_{ij}^*
&= {\mathcal R}_{ip} {\mathcal R}_{jq} Q_{pq}
 =Q_{ij} + \varepsilon_{jq} Q_{iq} + \varepsilon_{ip}Q_{pj} + \mathcal {O}(\varepsilon^2),  \\
Q_{ij,k}^*
&={\mathcal R}_{ip} {\mathcal R}_{jq} {\mathcal R}_{kr} Q_{pq, r}
 =Q_{ij,k} + \varepsilon_{kr} Q_{ij, r} + \varepsilon_{jq}Q_{iq ,k} + \varepsilon_{ip} Q_{pj ,k}
  + \mathcal {O}(\varepsilon^2).
\end{align*}
Thus, we had approximately that
\begin{align*}
0&=s(Q_{ij}^*, Q_{ij,k}^*, \cdots) - s(Q_{ij}, Q_{ij,k}, \cdots)
 \approx \frac{\partial s}{\partial Q_{ij}} (Q_{ij}^* - Q_{ij}) + \frac{\partial s}{\partial Q_{ij, k}} (Q_{ij, k}^* - Q_{ij, k}) \\
 &=\frac{\partial s}{\partial Q_{ij}}
 (\varepsilon_{jq} Q_{iq} + \varepsilon_{ip}Q_{pj})
 + \frac{\partial s}{\partial Q_{ij, k}} ( \varepsilon_{kr} Q_{ij, r} + \varepsilon_{jq}Q_{iq ,k} + \varepsilon_{ip} Q_{pj ,k}  ) \\
 &=\big [
 2 Q_{ik} \frac{\partial s}{\partial Q_{jk}}
 +2 (\partial_k Q_{il}) \frac{\partial s}{\partial (\partial_k Q_{jl})}
 +( \partial_{i} Q_{kl})\frac{\partial s}{\partial (\partial_j Q_{kl})}
 \big ] \varepsilon_{ji}
,
\end{align*}
where we used the symmetry of $Q_{ij}=Q_{ji}$ in the last step to relabel indices. Therefore, $ [
 2 Q_{ik} \frac{\partial s}{\partial Q_{jk}}
 +2 (\partial_k Q_{il}) \frac{\partial s}{\partial (\partial_k Q_{jl})}
 +( \partial_{i} Q_{kl})\frac{\partial s}{\partial (\partial_j Q_{kl})}]$ must be symmetric due to arbitrariness of the anti-symmetric tensors $\varepsilon$.
 This completed the proof.
\end{proof}

\begin{lemma}\label{proof of the epr5 geq 0}
(A Proof of the sufficient condition \eqref{epr5 geq 0} for $epr5 \geq 0$)
\end{lemma}
\begin{proof}
Assuming $\bar \alpha_8=0, \bar \alpha_7 \geq 0$ and denoting the residual terms
\begin{align} \label{epr5'}
epr5'=
&  \bar \alpha_4 |A|^2 + \mu_2 Y:A  + (\bar \alpha_5 +\bar \alpha_6)P:A + \mu_1 |Y|^2 +\bar \alpha_1 |P|^2,
\end{align}
as a quadratic form of tensors $A$, $Y$ and $P \equiv Q \cdot A$.
Without loss of generality, we firstly chose $A$ as the principle element.
By letting $\bar \alpha_4 \neq 0$, we had
\begin{align*}
epr5'=
&  \bar \alpha_4 \big |A + \frac{\mu_2}{2 \bar \alpha_4}Y + \frac{\bar \alpha_5 +\bar \alpha_6}{2 \bar \alpha_4}P    \big |^2 \\
& \underline{
+ ( \mu_1- \frac{\mu_2^2}{4 \bar \alpha_4} )|Y|^2
- \frac{\mu_2 (\bar \alpha_5 +\bar \alpha_6) }{2 \bar \alpha_4} Y:P
+ [ \bar \alpha_1 - \frac{(\bar \alpha_5 +\bar \alpha_6)^2 }{4 \bar \alpha_4}] |P|^2
}
\equiv I_1 +I_2 ,
\end{align*}
where the symmetry property of $A$ (or $Y$) was used to guarantee the commutativity of contractions, $P^T:A=P^T:A^T= P:A$ (or $Y: P^T= Y:P$).
Supposing $4 \mu_1 \bar \alpha_4  \neq \mu_2^2$, the second part of $epr5'$ became
\begin{align*}
I_2=
&  ( \mu_1- \frac{\mu_2^2}{4 \bar \alpha_4} )|Y|^2
- \frac{\mu_2 (\bar \alpha_5 +\bar \alpha_6) }{2 \bar \alpha_4} Y:P
+ [ \bar \alpha_1 - \frac{(\bar \alpha_5 +\bar \alpha_6)^2 }{4 \bar \alpha_4}] |P|^2 \\
&=   \frac{4 \mu_1 \bar \alpha_4 -  \mu_2^2}{4 \bar \alpha_4}
\big |Y    - \frac{\mu_2 (\bar \alpha_5 +\bar \alpha_6) }{4 \mu_1 \bar \alpha_4 -  \mu_2^2} P    \big |^2
+ [ \frac{4  \bar \alpha_1 \bar \alpha_4 - (\bar \alpha_5 +\bar \alpha_6)^2 }{4 \bar \alpha_4}
- \frac{{\mu_2^2 (\bar \alpha_5 +\bar \alpha_6)^2 }}{4 \bar \alpha_4({4 \mu_1 \bar \alpha_4 -  \mu_2^2})}
] |P|^2 \\
&\equiv I_{21} + I_{22}
.
\end{align*}
Collecting the quadratic terms together, the constraint $epr5' \geq 0$ yielded that
\begin{align*}
\bar \alpha_4 \geq 0, \quad
\frac{4 \mu_1 \bar \alpha_4 -  \mu_2^2}{4 \bar \alpha_4} \geq 0, \quad
[ \frac{4  \bar \alpha_1 \bar \alpha_4 - (\bar \alpha_5 +\bar \alpha_6)^2 }{4 \bar \alpha_4}
- \frac{{\mu_2^2 (\bar \alpha_5 +\bar \alpha_6)^2 }}{4 \bar \alpha_4({4 \mu_1 \bar \alpha_4 -  \mu_2^2})}
] \geq 0,
\end{align*}
which deduced that
\begin{equation*}
\begin{split}
&
\quad \bar \alpha_4 \geq 0,
\quad 4 \bar \alpha_4 \mu_1 - {\mu_2^2} \geq 0,
\quad 4 \bar \alpha_1 \bar \alpha_4 \mu_1 - \bar \alpha_1 \mu_2^2 - \mu_1({\bar \alpha_5 +\bar \alpha_6})^2 \geq 0.
\end{split}
\end{equation*}
Repeating above procedures, by choosing principal elements as $Y$ and $P$ separately, we obtained a complete sufficient condition for $epr5 \geq 0$ as
\begin{equation} \label{epr5 geq 00}
\begin{split}
&
\bar \alpha_8=0, \quad \bar \alpha_7 \geq 0,
\quad \bar \alpha_4 \geq 0,
\quad \mu_1 \geq 0,
\quad \bar \alpha_1 \geq 0, \\
&
\quad 4 \bar \alpha_4 \mu_1 - {\mu_2^2} \geq 0,
\quad 4 \bar \alpha_4 \bar \alpha_1   - ({\bar \alpha_5 +\bar \alpha_6})^2 \geq 0, \\
&\quad 4 \bar \alpha_1 \bar \alpha_4 \mu_1 - \bar \alpha_1 \mu_2^2 - \mu_1({\bar \alpha_5 +\bar \alpha_6})^2 \geq 0.
\end{split}
\end{equation}
This completed the proof.
\end{proof}

An alternative proof could be given by writing down the symmetric matrix associated with the quadratic form \eqref{epr5'} of $(A, Y, P)$ as
\begin{align*}
\begin{pmatrix}
\bar \alpha_4 & \frac{\mu_2}{2} & \frac{\bar \alpha_5 + \bar \alpha_6}{2}
\\  & \mu_1 & 0  \\
* & & \bar \alpha_1
\end{pmatrix}.
\end{align*}

\section{Appendix B: A Brief Discussion on the Angular Momentum Conservation of Liquid Crystals}
The conservations laws of mass, momentum, angular momentum and total energy played an important role in our mathematical modeling. However, unlike the mass, momentum and energy conservation laws, the conservation of angular momentum was not fully appreciated.

In this section, we presented a brief discussion on the angular momentum equations, and showed that the evolution equation of $w$ in \eqref{angmom} was a particular solution of the approximate equations for angular momentum conservation. For nematic liquid crystals, the classical conservation law of angular momentum took the following local Eulerian form
\begin{align}
&\frac{\partial}{\partial t}(\rho r \times v + \rho_1 d \times w)  + \nabla \cdot [ v \otimes (\rho r \times v + \rho_1 d \times w)]
=r \times \xi + \nabla \cdot (r \times \sigma + d \times \pi),\label{angu_append}
\end{align}
where $\sigma$ and $\pi$ were flow stress tensor and director surface torque, respectively. By applying the continuity equation, it became
\begin{align*}
&\epsilon_{ijk}[\rho \frac{d}{dt}( { {r_j v_k}}) + \rho_1 \frac{d}{dt}({ {d_j w_k}})] - \epsilon_{ijk} r_j \xi_k - \epsilon_{ijk} \partial_l (r_j \sigma_{lk} + d_j \pi_{lk}) \\
&=\epsilon_{ijk} (\rho  r_j \dot {v}_k + \rho_1 d_j \dot {w}_k) - \epsilon_{ijk} (r_j \xi_k + \sigma_{jk} + r_j \sigma_{lk,l} + d_{j,l} \pi_{lk} + d_{j} \pi_{lk,l})\\
&=\epsilon_{ijk} r_j(\rho \dot {v}_k -\xi_k - \sigma_{lk,l}) + \epsilon_{ijk} ( \rho_1 d_j \dot {w}_k - d_{j} \pi_{lk,l} - d_{j,l} \pi_{lk} -  \sigma_{jk})\\
&= \epsilon_{ijk} [  d_j  ( \rho_1 \dot {w}_k - \pi_{lk,l})  - (d_{j,l} \pi_{lk} + \sigma_{jk})] = 0
,
\end{align*}
in the Lagrangian form. 
Here we used $\epsilon_{ijk}  v_j {v_k}=\epsilon_{ikj}  v_k {v_j} = 0$ and $\epsilon_{ijk}  w_j {w_k}=\epsilon_{ikj}  w_k {w_j} = 0$ in the second step, as well as the momentum conservation in the third step. In the case when $\epsilon_{ijk} (d_{j,l} \pi_{lk} + \sigma_{jk})$ could be approximated by $\epsilon_{ijk} d_{j} g_k$ for a vector function $g$, it deduced that $\epsilon_{ijk}  d_j  ( \rho_1 \dot {w}_k - \pi_{lk,l} - g_{k}) = 0$. A particular solution of the approximate equation was given by
\begin{align*}
\rho_1 \dot {w}_k - \pi_{lk,l} - g_{k}=0,
\end{align*}
which was also recognized as the Oseen's equation in literature\cite{Chandrasekhar1992Liquid, Oseen1933The}. In what follows, we showed that this approach could also be extended to the tensorial theory and leaded to a generalized hydrodynamic equation for liquid crystals.

As to the tensorial theory for liquid crystals, the generalized conservation law of angular momentum became
\begin{align}\label{angu_append_Q}
&\frac{\partial}{\partial t}(\rho r \times v + \rho_1 Q \dot \times \dot Q)  + \nabla \cdot [ v \otimes (\rho r \times v + \rho_1 Q \dot\times \dot Q)]
=r \times \xi + \nabla \cdot (r \times \sigma + Q \dot \times \pi),
\end{align}
where $Q \dot\times \dot Q= \epsilon_{ijk} Q_{jm} \dot Q_{km}$, $Q \dot \times \pi=\epsilon_{ijk} Q_{jm} \pi_{lkm}$, the third-order tensor $\pi_{lkm}$ was the surface torque of liquid crystals. By substituting the mass and momentum conservations into \eqref{angu_append_Q}, we obtained
\begin{align*}
&\epsilon_{ijk}[\rho \frac{d}{dt} ( { {r_j v_k}}) + \rho_1 \frac{d}{dt} ({ {Q_{jm} \dot Q_{km}}})] - \epsilon_{ijk} r_j \xi_k -\epsilon_{ijk} \partial_l (r_j \sigma_{lk} + Q_{jm} \pi_{lkm} ) \\
&=\epsilon_{ijk} (\rho  r_j \dot {v}_k + \rho_1 Q_{jm} \ddot {Q}_{km}) - \epsilon_{ijk} (r_j \xi_k + \sigma_{jk} + r_j \sigma_{lk,l} + Q_{jm}  \pi_{lkm, l} + Q_{jm, l} \pi_{lkm})\\
&=\epsilon_{ijk} r_j(\rho \dot {v}_k - \xi_k - \sigma_{lk,l}) + \epsilon_{ijk} ( \rho_1 Q_{jm} \ddot {Q}_{km} - Q_{jm} \pi_{lkm, l} -  Q_{jm,l} \pi_{lkm}  -  \sigma_{jk})\\
&=  \epsilon_{ijk} Q_{jm}( \rho_1 \ddot {Q}_{km} -  \pi_{lkm, l})  - \epsilon_{ijk} ( Q_{jm,l} \pi_{lkm} + \sigma_{jk})
 = 0
,
\end{align*}
where the formula $\rho \dot {v}_k = \xi_k + \sigma_{lk,l}$ was used in the third step. By making a similar approximation that $\epsilon_{ijk} ( Q_{jm,l} \pi_{lkm} + \sigma_{jk}) = \epsilon_{ijk} Q_{jm} g_{km}$, we had $ \epsilon_{ijk} Q_{jm}( \rho_1 \ddot {Q}_{km}  - g_{km} -  \pi_{lkm, l} ) = 0$, where the second-order tensor $(g + \nabla \cdot \pi)$ was also known as the molecular field of liquid crystals. A particular class of solutions was
\begin{equation} \label{Q-tensor angular comp}
\rho_1 \ddot {Q}_{km} = g_{km} + \partial_l (\pi_{l km}) ,
\end{equation}
which was employed in Eq. \eqref{Qangmom} as the starting point of our derivation. Clearly, Eq. \eqref{Q-tensor angular comp} was different from the one used by Qian and Sheng \cite{Qian1998Generalized}, in which only the body torque $g$ (or molecular field $h$ in Eq. (3) in ref. \cite{Qian1998Generalized}) was included, while the surface torque $\pi$ was neglected without justification.

\bibliographystyle{ieeetr}
\bibliography{ref}

\begin{thebibliography}{10}

\bibitem{deGennes1995The}
P.~G. de~Gennes and J.~Prost, {\em The Physics of Liquid Crystals}.
\newblock Oxford: Oxford University Press, 1995.

\bibitem{Han2015From}
J.~Han, Y.~Luo, W.~Wang, P.~Zhang, and Z.~Zhang, ``From microscopic theory to
  macroscopic theory: a systematic study on modeling for liquid crystals,''
  {\em Archive for Rational Mechanics and Analysis}, vol.~215, no.~3,
  pp.~741--809, 2015.

\bibitem{Wang2015Rigorous}
W.~Wang, P.~Zhang, and Z.~Zhang, ``Rigorous derivation from landau-de gennes
  theory to ericksen-leslie theory,'' {\em SIAM Journal on Mathematical
  Analysis}, vol.~47, no.~1, pp.~127--158, 2015.

\bibitem{Ericksen1962Hydrostatic}
J.~L. Ericksen, ``Hydrostatic theory of liquid crystals,'' {\em Archive for
  Rational Mechanics and Analysis}, vol.~9, no.~1, pp.~371--378, 1962.

\bibitem{Leslie1968Some}
F.~M. Leslie, ``Some constitutive equations for liquid crystals,'' {\em Archive
  for Rational Mechanics and Analysis}, vol.~28, no.~4, pp.~265--283, 1968.

\bibitem{Leslie1979Theory}
F.~M. Leslie, ``Theory of flow phenomena in liquid crystals,'' {\em Advances in
  Liquid Crystals}, vol.~4, pp.~1--81, 1979.

\bibitem{Chandrasekhar1992Liquid}
S.~Chandrasekhar, {\em Liquid crystals}.
\newblock Cambridge: Cambridge University Press, 1992.

\bibitem{Lin1995Nonparabolic}
F.-H. Lin and C.~Liu, ``Nonparabolic dissipative systems modeling the flow of
  liquid crystals,'' {\em Communications on Pure and Applied Mathematics},
  vol.~48, no.~5, pp.~501--537, 1995.

\bibitem{Lin2000Existence}
F.-H. Lin and C.~Liu, ``Existence of solutions for the ericksen-leslie
  system,'' {\em Archive for Rational Mechanics and Analysis}, vol.~154, no.~2,
  pp.~135--156, 2000.

\bibitem{Sun2008On}
H.~Sun and C.~Liu, ``On energetic variational approaches in modeling the
  nematic liquid crystal flows,'' {\em Discrete and Continuous Dynamical
  Systems}, vol.~23, no.~1-2, pp.~455--475, 2008.

\bibitem{Peng2018Conservation}
L.~Peng, Y.~Hu, and L.~Hong, ``Conservation-dissipation formalism for soft
  matter physics: I. equivalence with doi's variational approach,'' {\em
  submitted}, 2018.

\bibitem{Lin2001Static}
F.-H. Lin and C.~Liu, ``Static and dynamic theories of liquid crystals,'' {\em
  J. Partial Differential Equations}, vol.~14, no.~4, pp.~289--330, 2001.

\bibitem{Fanghua2014Recent}
F.-H. Lin and C.~Wang, ``Recent developments of analysis for hydrodynamic flow
  of nematic liquid crystals,'' {\em Philosophical Transactions of the Royal
  Society A Mathematical Physical and Engineering Sciences}, vol.~372,
  no.~2029, 2014.

\bibitem{Hieber2016Thermodynamical}
M.~Hieber and J.~Pr\"{u}ss, {\em Thermodynamical Consistent Modeling and
  Analysis of Nematic Liquid Crystal Flows}, vol.~183.
\newblock Springer, Tokyo: Mathematical Fluid Dynamics, Present and Future.
  Springer Proceedings in Mathematics and Statistics, 2016.

\bibitem{Hieber2017Dynamics}
M.~Hieber and J.~Pr\"{u}ss, ``Dynamics of the ericksen–leslie equations with
  general leslie stress i: the incompressible isotropic case,'' {\em
  Mathematische Annalen}, vol.~369, no.~3-4, pp.~977--996, 2017.

\bibitem{Feireisl2011On}
E.~Feireisl, E.~Rocca, and G.~Schimperna, ``On a non-isothermal model for
  nematic liquid crystals,'' {\em Nonlinearity}, vol.~24, no.~1,
  pp.~243--257(15), 2011.

\bibitem{Anna2017Non}
F.~D. Anna and C.~Liu, ``Non-isothermal general ericksen–leslie system:
  Derivation, analysis and thermodynamic consistency,'' {\em Archive for
  Rational Mechanics and Analysis}, pp.~1--81, 2017.

\bibitem{Gay2011The}
F.~Gay-Balmaz and C.~Tronci, ``The helicity and vorticity of liquid-crystal
  flows,'' {\em Proceedings Mathematical Physical and Engineering Sciences},
  vol.~467, no.~2128, pp.~1197--1213, 2011.

\bibitem{Beris1994Thermodynamics}
A.~N. Beris and B.~J. Edwards, ``Thermodynamics of flowing systems: with
  internal microstructure,'' {\em Oxford University Press on Demand}, 1994.

\bibitem{Qian1998Generalized}
T.~Qian and P.~Sheng, ``Generalized hydrodynamic equations for nematic liquid
  crystals,'' {\em Physical Review E}, vol.~58, no.~58, pp.~7475--7485, 1998.

\bibitem{Feireisl2014Evolution}
E.~Feireisl, E.~Rocca, G.~Schimperna, and A.~Zarnescu, ``Evolution of
  non-isothermal landau-de gennes nematic liquid crystals flows with singular
  potential,'' {\em Communications in Mathematical Sciences}, vol.~12, no.~2,
  pp.~317--343, 2014.

\bibitem{Feireisl2015Nonisothermal}
E.~Feireisl, G.~Schimperna, E.~Rocca, and A.~Zarnescu, ``Nonisothermal nematic
  liquid crystal flows with the ball–majumdar free energy,'' {\em Annali di
  Matematica Pura ed Applicata (1923 -)}, vol.~194, no.~5, pp.~1269--1299,
  2015.

\bibitem{Sonnet2004Continuum}
A.~M. Sonnet, P.~L. Maffettone, and E.~G. Virga, ``Continuum theory for nematic
  liquid crystals with tensorial order,'' {\em Journal of Non-Newtonian Fluid
  Mechanics}, vol.~119, no.~1, pp.~51--59, 2004.

\bibitem{Stewart2004The}
I.~W. Stewart, {\em The Static and Dynamic Continuum Theory of Liquid Crystals:
  A Mathematical Introduction}.
\newblock London and New York: Taylor and Francis, 2004.

\bibitem{Flory1984Molecular}
P.~J. Flory, ``Molecular theory of liquid crystals,'' {\em Advances in Polymer
  Science}, vol.~59, no.~12, pp.~1--36, 1984.

\bibitem{Lekkerkerker1984On}
H.~N.~W. Lekkerkerker, P.~Coulon, R.~V.~D. Haegen, and R.~Deblieck, ``On the
  isotropic‐liquid crystal phase separation in a solution of rodlike
  particles of different lengths,'' {\em Journal of Chemical Physics}, vol.~80,
  no.~7, pp.~3427--3433, 1984.

\bibitem{He2016Isotropic}
H.~He, E.~M. Sevick, and D.~R.~M. Williams, ``Isotropic and nematic liquid
  crystalline phases of adaptive rotaxanes,'' {\em Journal of Chemical
  Physics}, vol.~144, no.~12, pp.~2186--2199, 2016.

\bibitem{jou1996extended}
D.~Jou, J.~Casas-V{\'a}zquez, and G.~Lebon, {\em Extended irreversible
  thermodynamics}.
\newblock Springer, 1996.

\bibitem{zhu2015conservation}
Y.~Zhu, L.~Hong, Z.~Yang, and W.-A. Yong, ``Conservation-dissipation formalism
  of irreversible thermodynamics,'' {\em Journal of Non-Equilibrium
  Thermodynamics}, vol.~40, no.~2, pp.~67--74, 2015.

\bibitem{Ericksen1961Conservation}
J.~L. Ericksen, ``Conservation laws for liquid crystals,'' {\em Transactions of
  the Society of Rheology}, vol.~5, no.~1, pp.~23--34, 1961.

\bibitem{Frank1958I}
F.~C. Frank, ``I. liquid crystals. on the theory of liquid crystals,'' {\em
  Discussions of the Faraday Society}, vol.~25, no.~25, pp.~19--28, 1958.

\bibitem{Ball2017Mathematics}
J.~M. Ball, ``Mathematics and liquid crystals,'' {\em Molecular Crystals and
  Liquid Crystals}, vol.~647, no.~1, pp.~1--27, 2017.

\bibitem{Durand1969Quasielastic}
G.~Durand, L.~Leger, F.~Rondelez, and M.~Veyssie, ``Quasielastic rayleigh
  scattering in nematic liquid crystals.,'' {\em Physical Review Letters},
  vol.~23, no.~25, pp.~1361--1363, 1969.

\bibitem{Lee1986Computations}
S.~Lee and R.~B. Meyer, ``Computations of the phase equilibrium, elastic
  constants, and viscosities of a hard‐rod nematic liquid crystal,'' {\em
  Journal of Chemical Physics}, vol.~84, no.~6, pp.~3443--3448, 1986.

\bibitem{Allen1988Calculation}
M.~P. Allen and D.~Frenkel, ``Calculation of liquid-crystal frank constants by
  computer simulation,'' {\em Physical Review A}, vol.~37, no.~5,
  pp.~1813--1816, 1988.

\bibitem{Edwards1990Generalized}
B.~J. Edwards, ``Generalized constitutive equation for polymeric liquid
  crystals: Part 2. non-homogeneous systems,'' {\em Journal of Non-Newtonian
  Fluid Mechanics}, vol.~36, no.~36, pp.~243--254, 1990.

\bibitem{Grmela2018Generic}
M.~Grmela, ``Generic guide to the multiscale dynamics and thermodynamics,''
  {\em Journal of Physics Communications}, vol.~2, no.~3, p.~032001, 2018.

\bibitem{de2013non}
S.~R. de~Groot and P.~Mazur, {\em Non-equilibrium thermodynamics}.
\newblock Amsterdam: North-Holland Publishing Company, 1962.

\bibitem{doi2013soft}
M.~Doi, {\em Soft matter physics}.
\newblock Oxford: Oxford University Press, 2013.

\bibitem{Oseen1933The}
C.~W. Oseen, ``The theory of liquid crystals,'' {\em Transactions of the
  Faraday Society}, vol.~29, no.~140, pp.~883--899, 1933,.

\end{thebibliography}
\end{document}